\documentclass[journal,10pt,twocolumn,twoside]{IEEEtran}
\ifCLASSINFOpdf
\else
\fi

\usepackage{graphics} % for pdf, bitmapped graphics files
\usepackage{graphicx}

\usepackage{times} % assumes new font selection scheme installed
\usepackage{amsmath} % assumes amsmath package installed
\usepackage{amssymb}  % assumes amsmath package installed

\usepackage{tikz} % <--- if tikz pictures are used
\usetikzlibrary{arrows}				
\usetikzlibrary{shapes.misc}
\usetikzlibrary{patterns}
\usetikzlibrary{fit}
\usepackage{pgfplots}
\usetikzlibrary{calc}
\usetikzlibrary{decorations.markings}
\pgfplotsset{every tick label/.append style={font=\small}}
\usepgfplotslibrary{fillbetween}
\usetikzlibrary{decorations.pathreplacing}
\usetikzlibrary{calc, shadings, shadows, snakes, shapes, arrows,trees,patterns}
\usetikzlibrary{chains, fit, automata, positioning}
\usetikzlibrary{math}
\usepackage{pgfplots} % <--- if pgfplots is used

\usepackage{algpseudocode}%
\usepackage[ruled,vlined,linesnumbered]{algorithm2e}
\usepackage{epstopdf}
\usepackage{cases}
\usepackage{stackrel}

\usepackage{color}

\usepackage{lipsum}% http://ctan.org/pkg/lipsum
\usepackage{verbatim}

\usepackage{amsthm}
\usepackage{siunitx}
\usepackage{mathtools}
\usepackage{microtype}
\allowdisplaybreaks

\newtheorem{theorem}{\bf Theorem} \newtheorem{definition}[theorem]{\bf Definition} 
\newtheorem{lemma}[theorem]{\bf Lemma} \newtheorem{remark}[theorem]{\bf Remark}
 \newtheorem{proposition}[theorem]{\bf Proposition} 
\newtheorem{assumption}[theorem]{\bf Assumption}  
\newtheorem{Algorithm}[theorem]{\bf Algorithm}

\newcommand{\X}{\mathbb{X}}
\newcommand{\U}{\mathbb{U}}

\newcommand{\R}{\mathbb{R}}
\newcommand{\W}{\mathbb{W}}
\newcommand{\V}{\mathbb{V}}
\newcommand{\N}{\mathbb{N}}

\colorlet{istorange}{orange}
\colorlet{istgreen}{green!50!black}
\colorlet{istblue}{blue} % use structure theme to change
\colorlet{istred}{red!90!black}

\title{\LARGE \bf Uncertainties and output feedback in rollout event-triggered control$^*$}

\author{Stefan Wildhagen and Frank Allg{\"o}wer% <-this % stops a space
	\thanks{$^*$Funded by Deutsche Forschungsgemeinschaft (DFG, German Research Foundation) under Germany's Excellence Strategy - EXC 2075 - 390740016 and under grant AL 316/13-2 - 285825138. The authors are with the University of Stuttgart, Institute for Systems Theory and Automatic Control, Germany.
		{\tt\small \{wildhagen}{\tt\small,allgower\}@ist.uni}{\tt\small -stuttgart.de}. Corresponding author: Stefan Wildhagen.}  %
}

\begin{document}

	\maketitle
	\thispagestyle{empty}
	\pagestyle{empty}

	%%%%%%%%%%%%%%%%%%%%%%%%%%%%%%%%%%%%%%%%%%%%%%%%%%%%%%%%%%%%%%%%%%%%%%%%%%%%%%%%
	\begin{abstract}
		The fact that event-triggered control (ETC) often exhibits an improved performance-communication tradeoff over time-triggered control renders it especially useful for Networked Control Systems (NCSs). However, it has proven difficult to characterize the traffic produced by ETC a priori. Rollout ETC addresses this issue by using a triggering and control law that is \emph{implicitly} defined by the solution to an optimal control problem (OCP), instead of an explicit one as in classical ETC. This allows to directly incorporate predefined constraints on the transmission traffic as well as on states and inputs. In this article, we examine the practically relevant case when output instead of state measurements are available, and measurements as well as the LTI plant are subject to uncertainties. To address these challenges, we adapt methods from robust tube-based model predictive control and propose three different strategies to implement an error feedback in an NCSs setup, the applicability of which depends on the capabilities of the actuator. We establish recursive feasibility, robust constraint satisfaction and convergence. Finally, we illustrate our results in a numerical example.
	\end{abstract}
	
	
	%%%%%%%%%%%%%%%%%%%%%%%%%%%%%%%%%%%%%%%%%%%%%%%%%%%%%%%%%%%%%%%%%%%%%%%%%%%%%%%%
	
	\section{Introduction}
	
	\IEEEPARstart{E}{vent}-triggered control (ETC) (see \cite{Heemels12} for an overview) has become popular in recent years as a strategy to tackle the challenges posed by Networked Control Systems (NCSs). The main paradigm of ETC is that a state-dependent trigger condition is monitored during runtime to determine when a control update should be transmitted. If well-designed, ETC exhibits a better performance-communication tradeoff than time-triggered control (TTC), such that communication is reduced. Therefore, so a common argument, ETC avoids congestion of the communication medium and thus leads to low delays and packet dropout probabilities.
	
	However, it is in fact hard to characterize the transmission traffic pattern produced by a given ETC system a priori, as remarked in the recent work \cite{postoyan2019inter}, ``very little is known about the actual behavior of the inter-event times.'' Only recently, there have been several works on this matter, but the results obtained thus far consider special system classes and trigger rules \cite{postoyan2019inter,gleizer2020trafficmodels} or suffer from high computational complexity \cite{kolarijani2016formal,delimpaltadakis2020traffic}. An especially challenging task is to \emph{design} ETC mechanisms which fulfill predefined requirements on the transmission traffic, e.g., a certain bandwidth limit. Achieving this typically requires numerous iterations of numerical simulation or formal traffic pattern analysis, and a subsequent redesign of the trigger and control law. Consequently, a guarantee that ETC avoids congestion is in fact not straightforward and requires great attention in the design process.
	
	Rollout ETC \cite{Antunes14} has the potential to facilitate this design process, due to its ability to directly incorporate bandwidth constraints in the trigger condition. In contrast to classical ETC, where the trigger and control law is given in an \emph{explicit} manner, the trigger and control law in rollout ETC is \emph{implicitly} defined by the solution to an optimal control problem (OCP), maximizing the expected performance over a prediction horizon w.r.t. the triggering time instants and control updates. In rollout ETC, the state is measured periodically, the OCP is solved based on the new measurement, and then the first part of the optimized transmission schedule and input trajectory are applied in closed loop. The design of rollout ETC schemes which respect a certain bandwidth is simple: it is only required to include the desired bandwidth constraint in the OCP, and then the optimal transmission strategy is determined by the implicit trigger rule automatically. Similarly, explicit constraints on the state and input can be incorporated, which is not readily possible in TTC and classical ETC. The main idea of rollout ETC was introduced in \cite{Antunes14} for linear and unconstrained plants, and was later extended to constrained \cite{Gommans17} and nonlinear plants \cite{Wildhagen19_2} using tools from model predictive control (MPC).
	
	While uncertainty as well as output feedback in classical ETC are well-understood (see, e.g., \cite{Donkers12,Heemels13_2,Liu15,Abdelrahim17}), these issues have found little attention in the literature on rollout ETC so far. Nonetheless, they are very important for practical applicability: often, only output measurements are available instead of full state information, measurement noise arises in all applications to a greater or lesser extent, and disturbances on the plant occur due to physical disturbances, unmodeled dynamics or parameter uncertainties. At the same time, real-world plants are often subject to state and input constraints, which make safe operation under uncertainty all the more challenging. Some preliminary work on this set of issues was presented in \cite{Antunes14}, which investigated state-feedback rollout ETC for unconstrained plants subject to additive, stochastic disturbances.
	
	In this article, we qualify rollout ETC for output feedback, measurement noise and disturbances in the presence of state and input constraints. In particular, we consider an NCS where the controller is collocated with the sensor, and the controller is connected with the actuators via a bandwidth-constrained network (as in \cite{postoyan2019inter}-\cite{Liu15}). The linear time-invariant (LTI) plant is subject to bounded additive disturbances and state and input constraints, while only a noisy output of the plant is available for control. We exploit that rollout ETC is closely related to MPC, and build on existing robust tube-based MPC approaches (see, e.g., \cite{mayne2006robust_output}): First, we introduce $[1,H]$ robust control invariant (RCI) sets for NCSs, which bound the effect of the uncertainties despite the fact that transmissions might be up to $H$ time steps apart. We propose three different strategies to implement an error feedback in the considered NCS setup, each requiring different assumptions on the actuator and giving rise to a different form of the $[1,H]$ RCI set. Second, we introduce a schedule constraint in the prediction to ensure that inter-transmission intervals in closed loop are indeed no longer than $H$. Third, we propose a novel tube MPC scheme, in which the initial condition in the prediction is an optimization variable only in case there is also a transmission scheduled. We prove recursive feasibility of the proposed scheme, satisfaction of the state and input constraints despite the uncertainties and convergence to a bounded set.
	
	In our preliminary work on robust rollout ETC \cite{wildhagen2020robustrollout}, we considered a disturbance on the plant, but assumed perfect state measurements. Furthermore, we considered a zero-order hold (ZOH) actuator only. In this article, we include also output feedback and measurement noise and, in addition, treat two types of more capable actuators that are able to implement a more advanced error feedback. Further, we derive a checkable condition for existence of a $[1,H]$ RCI set for the ZOH actuator, and present an algorithm to construct this set.
	
	The remainder of this article is structured as follows. In Section \ref{sec:ncs_architecture}, we discuss the considered NCS architecture, present how to capture the bandwidth constraint, and introduce the implicit triggering and control law. In Section \ref{sec:bounding}, we bound the influence of the uncertainties and introduce the three different types of actuator. We detail the proposed robust rollout ETC algorithm in Section \ref{sec:MPC_scheme}, where we also establish recursive feasibility, robust constraint satisfaction and convergence guarantees. A numerical example is presented in Section \ref{sec:num_ex}, and a summary and outlook are given in Section \ref{sec:summary}. Some lengthy proofs are provided in the Appendix.
	
	$\N$ denotes the set of natural numbers, $\N_0\coloneqq \N\cup \{0\}$, $\N_{[a,b]}\coloneqq\N_0\cap[a,b]$ and $\N_{\ge a}\coloneqq\N_0\cap[a,\infty)$, $a,b\in\N_0$.  $I$ denote the identity matrix and $0$ the zero matrix of appropriate dimension. $A\succ0$ $(A\succeq 0)$ denotes a symmetric positive {(semi)}definite matrix. We call a compact and convex set containing the origin a $\mathcal{C}$ set. $\oplus$ denotes Minkowski set addition and $\ominus$ Pontryagin set difference. $\lor$, $\land$ denote the logical \emph{or}, \emph{and}, respectively.

	\section{NCS Architecture} \label{sec:ncs_architecture}
	
	We consider the NCS architecture depicted in Figure \ref{fig:network_tb}. The discrete-time LTI plant is subject to a disturbance $w_p$ and its measured output $y_p$ is affected by the measurement noise $v_p$. In our NCS setup, updates of the control input $u_c$ need to be transmitted to the actuator via a bandwidth-constrained communication network. This constraint is captured in a so-called token bucket traffic specification (TS) \cite{Tanenbaum11}, which restricts the time instants at which a transmission may be triggered and serves as a ``contract'' between the network and the control loop: if the assigned TS is respected by the transmission traffic, the network guarantees that there is no congestion such that dropout probabilities and delays are low \cite{zhang2012network}. The device running the rollout ETC is collocated with the sensor, which is why we speak of a ``smart sensor'' here. In each time step, it samples the output and solves an OCP to determine the control update $u_c$ and whether a transmission over the network should be triggered or not, as captured by $\gamma\in\{0,1\}$. At the other side of the network, an actuator applies the last received control update $u_c$ and an error feedback $u_e$ to the plant. Next, we detail each of the components before introducing a mathematical description of the overall NCS.
	
	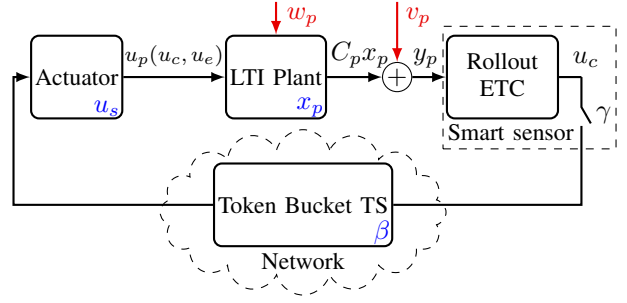
\begin{figure}
		\centering
		\begin{tikzpicture}[>= latex]
	%\tikzmath{\xmax = 7; \ymax =-2; \slength = 0.25;}
	\node [thick,draw,rectangle, inner sep=1.5pt,minimum size=1.1cm,align = center,rounded corners=3pt] (actuator) at (0.07*7,0) {\small{Actuator}};
	\node[anchor=south east,inner sep=1pt] at (actuator.south east) {\textcolor{istblue}{$u_s$}};
	\node [thick,draw,rectangle, inner sep=1.5pt,minimum size=1.1cm,align = center,rounded corners=3pt] (plant) at (0.42*7+0.2,0) {\small{LTI Plant}};
	\node[draw,circle,minimum size = 0.15cm,align=center,inner sep=0.3pt] (sum) at (0.65*7+0.2,0) {$+$};
	\node[anchor=south east,inner sep=1pt] at (plant.south east) {\textcolor{istblue}{$x_p$}};
	\node [thick,draw,rectangle, inner sep=1.5pt,minimum size = 1.1cm,align = center,rounded corners=3pt] (ctrl) at (0.85*7+0.2,0) {\small{\begin{tabular}{c} Rollout \\ ETC
			\end{tabular}}};
	\node [thin,dashed,draw,rectangle,minimum height = 1.6cm, minimum width = 2.3cm,align = left] (smartsensor) at (0.93*7,0.07*-2) {\hspace{0.0cm}\\[0.97cm] \hspace{-0.5cm}\small{Smart sensor}};
	\node [thin,dashed,draw,cloud, cloud puffs = 16,minimum height = 2.2cm, minimum width = 3.8cm,align = left] (smartsensor) at (0.5*7,-1.8) {};
	\node (netwname) at (0.5*7,-2.45) {\small{Network}};
	\node [thick,draw,rectangle, inner sep=1.5pt,minimum size = 1.1cm,align = center,rounded corners=3pt] (netw) at (0.5*7,-1.7) {\small{Token Bucket TS}};
	\node[anchor=south east,inner sep=1pt] at (netw.south east) {\textcolor{istblue}{$\beta$}};
	
	\node (upper) at (7.2,0.2*-2){};
	\node (lower) at (7.2,0.2*-2-0.25){};
	\node (lower2) at (7.2+0.25*0.5,0.2*-2 -0.25*0.866){};
	\node (dist) at (0.47*7+0.2,0.8) {\textcolor{istred}{$w_p$}};
	\node (noise) at (0.69*7+0.2,0.8) {\textcolor{istred}{$v_p$}};
	
	\path (0,0) --++ (7.2,0) node(helpne){} --++ (0,-1.7) node(helpse){} --++ (-7.2-0.05*7,0) node(helpsw){} --++ (0,1.7) node(helpnw){};
	
	\draw [->,thick] (plant.east) -- (sum.west) node[pos = 0.6,above] {$C_px_p$};
	\draw [->,thick] (sum.east) -- (ctrl.west) node[pos = 0.4,above] {$y_p$};
	\draw [thick] (ctrl.east) -- (helpne.center)node[pos=1.2,above] {$u_c$};
	\draw [thick] (helpne.center) -- (upper.center) -- (lower2.center) node[pos=0.5,right] {$\gamma$};
	\draw[thick] (lower.center) -- (helpse.center) -- (netw.east);
	\draw[->,thick] (netw.west) -- (helpsw.center) -- (helpnw.center) -- (actuator.west);
	\draw[->,thick] (actuator.east) -- (plant.west) node[pos = 0.5,above] {\footnotesize$u_p(u_c,u_e)$};
	
	\draw[->,thick,color=istred] (0.42*7+0.2,1) -- (plant.north);
	\draw[->,thick,color=istred] (0.65*7+0.2,1) -- (sum.north);
	\end{tikzpicture}
		\caption{Considered configuration of the NCS with token bucket TS.}
		\label{fig:network_tb}
	\end{figure}
	
	\subsection{Plant}
	
	The plant is given by the perturbed discrete-time LTI system
	\begin{align}
	x_p(k+1) &= A_px_p(k)+B_pu_p(k)+w_p(k), x_p(0)=x_{p,0}, \nonumber \\
	y_p(k) &= C_p x_p(k) + v_p(k),  \label{eq:plant}
	\end{align}
	where $k\in\N_{0}$ denotes time, $x_p(k)\in\R^{n_p}$ the plant state, $u_p(k)\in\R^{m_p}$ the plant input, and $y_p(k)\in\R^{q_p}$ the measured plant output. Both state and input are subject to the constraints $x_p(k)\in\X_p$ and $u_p(k)\in\U_p$, where $\X_p\subseteq\R^{n_p}$ and $\U_p\subseteq\R^{m_p}$ are closed sets containing the origin. The variables $w_p(k)\in\R^{n_p}$ and $v_p(k)\in\R^{q_p}$ represent disturbances and measurement noise, respectively. They cannot be measured, but are known to be bounded
	\begin{equation*}
	w_p(k)\in \W_p, \quad v_p(k) \in \V_p,
	\end{equation*}
	where $\W_p\subseteq\R^{n_p}$ and $\V_p\subseteq\R^{q_p}$ are $\mathcal{C}$ sets. The pair $(A_p,B_p)$ is assumed to be stabilizable and $(A_p,C_p)$ detectable. To measure performance, a quadratic cost is associated with the plant
	\begin{equation}
	x_p^\top Q x_p + u_p^\top R u_p, \label{eq:cost_plant}
	\end{equation}
	where $Q\succ 0\in\R^{n_p\times n_p}$ and $R\succ 0\in\R^{m_p\times m_p}$.
	
	\subsection{Smart sensor}
	
	In each time step, the smart sensor measures the noisy output of the plant $y_p(k)$. It also runs the rollout ETC (RETC) which, based on the taken measurement, calculates the control update $u_c(k)$ and determines whether this update should be transmitted via the network to the actuator. This binary transmission decision takes the values
	\begin{equation*}
	\gamma(k) = \begin{cases}
	1 & \text{if a transmission is triggered at time }k \\
	0 & \text{if not}
	\end{cases}.
	\end{equation*}
	The control update and transmission decision are computed according to the law
	\begin{equation} \label{eq:impl_trig_contr_law}
	\begin{bmatrix}	u_c(k)^\top & \gamma(k)	\end{bmatrix}^\top = \kappa_{\text{RETC}}(y_p(k),k),
	\end{equation}
	where $\kappa_{\text{RETC}}:\R^p\times\N_0\to\R^m\times\{0,1\}$ is \emph{implicitly} defined by the solution to an OCP\footnote{We note that the triggering and control laws in \cite{Antunes14,Gommans17,Wildhagen19_2}, in the absence of disturbances, could be equivalently reformulated in the form \eqref{eq:impl_trig_contr_law} by using a cyclically time-varying prediction horizon as defined later in \eqref{eq:shrinking_horizon}.}. The implicit triggering and control law will be detailed in Section \ref{sec:MPC_scheme}.
	
	\subsection{Network and token bucket TS}
	
	As depicted in Figure \ref{fig:network_tb}, an update of the control input $u_c$ cannot be applied to the plant directly and must be transmitted to the actuator via a communication network. Throughout this paper, we assume that transmissions over the network must fulfill the token bucket TS \cite{Tanenbaum11}, which is based on an analogy to a bucket containing a variable amount of tokens. New tokens are added to the bucket at a constant rate of $g\in\N_{\ge 1}$, while triggering a transmission comprises a certain cost $c\in\N_{\ge g}$. The bucket has a maximum capacity $b\in\N_{\ge c}$ and arriving tokens are discarded if the bucket is already full. Hence, the bucket level $\beta$ is governed by the difference equation
	\begin{equation*}
	\beta(k+1)=\min\{\beta(k)+g-\gamma(k) c,b\}, \;\beta(0)=\beta_0.
	\end{equation*}
	The rollout ETC is only allowed to trigger a transmission (i.e.,  $\gamma(k)=1$) if the current amount of tokens is sufficient to support the transmission cost $c$, which can be described by the constraint $\beta(k)\ge 0$ for all $k\in\N_0$. Hence, the average transmission rate allowed by the token bucket TS is $\frac{g}{c}$. In addition, it can be guaranteed a priori that transmissions are possible at a base period of $\lceil\frac{c}{g}\rceil$ time instants. If the parameters $g$ and $c$ are chosen such that $\frac{g}{c}$ is at most as high as the bandwidth of the bottleneck link in the network, congestion is avoided, which in turn results in low packet dropout probabilities and delays \cite{zhang2012network}. This justifies to assume in this article that both are negligible if the token bucket TS is fulfilled.
	
	\subsection{Actuator}
	
	The actuator stores the last received control update in an auxiliary state $u_s$ according to
	\begin{equation*}
	u_s(k+1)=\gamma(k)u_c(k) + (1-\gamma(k))u_s(k)
	\end{equation*}
	with $u_s(0)=u_{s,0}$. Borrowing from the typical approach in tube-based MPC, the control update $u_c$ is appended with a so-called error feedback term $u_e$ in order to contain the error between a nominal version of the plant and the real, uncertain plant \eqref{eq:plant} in a bounded set. In Section \ref{sec:bounding}, we will detail the error feedback and see that for some types of actuators, it can be updated regardless of the transmission decision. For now, we presume that the error feedback at time $k$ is given by $u_e(k)$. Consequently, the input applied to the plant is
	\begin{equation*}
	u_p(k)=\gamma(k)u_c(k) + (1-\gamma(k))u_s(k) + u_e(k).
	\end{equation*}
	
	\subsection{Overall NCS}
	
	We have seen in the previous subsections that the actuator stores the last received control update in an auxiliary state $u_s$ and that the bucket level $\beta$ is a dynamical state of the NCS as well. As a result, in addition to the plant state $x_p$, these variables represent a component of the NCS's overall state
	\begin{equation*}
	x \coloneqq \begin{bmatrix} x_p^\top & u_s^\top & \beta	\end{bmatrix}^\top.
	\end{equation*}
	Likewise, the overall input consists of the control update $u_c$, the transmission decision $\gamma$ and the error feedback $u_e$,
	\begin{equation*}
	u \coloneqq \begin{bmatrix} u_c^\top & \gamma & u_e^\top \end{bmatrix}^\top.
	\end{equation*}
	The overall NCS dynamics then become
	\begin{align} \label{eq:system}
	&x(k+1) = f(x(k),u(k)) + w(k), \\
	&f(x,u) \hspace{-1pt}\coloneqq\hspace{-1.5pt} \begin{bmatrix} A_px_p + B_p((1-\gamma) u_s + \gamma u_c + u_e)\\ (1-\gamma) u_s + \gamma u_c \\ \min\{\beta+g-\gamma c,b\} \end{bmatrix}\hspace{-3pt}, \:
	w \hspace{-1pt}\coloneqq\hspace{-1.5pt} \begin{bmatrix} w_p \\ 0 \\ 0 \end{bmatrix} \nonumber
	\end{align}
	and the overall output equation is
	\begin{equation*}
	y(k) = Cx(k) + v(k), \; v \coloneqq \begin{bmatrix} v_p^\top & 0 & 0 \end{bmatrix}^\top, \nonumber
	\end{equation*}
	where $C\coloneqq\text{diag}\{C_p,I,1\}$. We collect all constraints in the overall state and input constraints
	$x(k)\in\X\coloneqq \X_p \times \U_p \times \N_{[0,b]}$ and $u(k)\in\U\coloneqq\{(u_c,\gamma,u_e)\in\R^{m_p\times 1\times m_p}|u_c+u_e\in\U_p,\gamma\in\{0,1\}\}$.
	The overall input constraint takes this form since the sum of updated control $u_c$ and error feedback $u_e$ is applied to the plant, such that their sum must be contained in $\U_p$. Further, we denote the dimension of the overall state, input and output as $n\coloneqq n_p+m_p+1$, $m\coloneqq 2m_p+1$ and $q\coloneqq q_p+m_p+1$, respectively, and we define the uncertainty sets $w(k)\in\W\coloneqq \W_p\times\{0\}\times\{0\}$ and $v(k)\in\V\coloneqq \V_p\times\{0\}\times\{0\}$.

	\section{Bounding the influence of the uncertainties} \label{sec:bounding}
	
	In this section, we bound the influence of the uncertainties in order to account for them in a so-called tube-based MPC scheme: There, the idea is to use a nominal version of the plant, neglecting the uncertainties, in the prediction and to keep the error between the real and nominal state in a tube through an error feedback. This allows to tighten the constraints used in the prediction by the tube size such that satisfaction of the original constraints is ensured. Here, we use a Luenberger observer to estimate the plant state and show that the estimation error is bounded by an invariant set. Then, we will bound the influence of the disturbance on the control error between the estimated state and the state of the nominal system used for predictions. To this end, we will introduce a novel definition of an error feedback and a $[1,H]$ RCI set tailored to NCSs, which can handle inter-transmission intervals of several time steps in length. We will present three distinct strategies to implement the error feedback in an NCS setup, depending on the capabilities of the actuator. Lastly, the invariant sets for estimation and control error will be combined to form the tube bounding the error between the real and nominal system.
	
	\subsection{Observer and bound for estimation error}
	
	We run a Luenberger observer at the smart sensor to estimate the state of the plant. The estimated plant state $\hat{x}_p(k)\in\R^{n_p}$ follows the dynamics
	\begin{align}
	\hat{x}_p(k+1) = A_p \hat{x}_p(k) + B_p u_p(k) + L_p(y_p(k)-C_p \hat{x}_p(k)), \label{eq:luen_dyn}
	\end{align}
	where $L_p$ is such that $A_p-L_pC_p$ is Schur. Plugging in $u_p(k)$, to which the observer is assumed to have access, we rewrite the observer dynamics \eqref{eq:luen_dyn} as
	\begin{align} \label{eq:luen_alt}
	\hat{x}_p(k+&1) = A_p \hat{x}_p(k) + B_p (1-\gamma(k))\hat{u}_s(k)+B_p\gamma(k)u_c(k) \nonumber \\ &+B_p u_e(k) + L_pC_p(x_p(k)-\hat{x}_p(k))+L_pv_p(k).
	\end{align}
	
	For the overall state $x$, we define an observer system which consists of the Luenberger observer \eqref{eq:luen_alt} for $x_p$, and simulators for $u_s$ and $\beta$ (with their simulated versions denoted $\hat{u}_s$ and $\hat{\beta}$). The latter two are internal variables of the controller, which is why they are known with certainty. The reason for defining this overall observer system nonetheless is that the real overall system \eqref{eq:system} contains these states as well. It reads
	\begin{equation} \label{eq:system_obs}
	\hat{x}(k+1) = f(\hat{x}(k),u(k)) + B_\epsilon \epsilon(k) + B_v v(k),
	\end{equation}
	where $\hat{x}(k)\coloneqq\begin{bmatrix} \hat{x}_p(k)^\top & \hat{u}_s(k)^\top & \hat{\beta}(k) \end{bmatrix}^\top\in\R^n$ is the observer state, $\epsilon(k)\coloneqq x(k)-\hat{x}(k)$ is the \emph{estimation error}, $B_\epsilon\coloneqq\text{diag}\{L_pC_p,0,0\}$ and $B_v\coloneqq\text{diag}\{L_p,0,0\}$. 

	Our goal now is to derive a bound on the estimation error $\epsilon$, to be able to account for it later in a tube-based scheme. Such a bound can be derived by means of a robust positively invariant (RPI) set, for which we recall the following definition.
	\begin{definition}[cf. \cite{kolmanovsky1998theory}]
		A set $\Psi\subseteq\R^{n}$ is RPI for the estimation error $\epsilon$ if and only if for all $x(k),\hat{x}(k)\in\X$ such that $\epsilon(k)\in\Psi$, for any $w(k)\in\W$, $v(k)\in\V$ it holds that $\epsilon(k+1)\in\Psi$.
	\end{definition}
	
	\begin{lemma} \label{lem:RPI_obs}
		Suppose there exist $L_p\in\R^{n_p\times q_p}$ and a set $\Psi_p\in\R^{n_p}$, containing the origin, which satisfy
		\begin{equation} \label{eq:cond_RPI_obs}
		(A_p-L_pC_p)\Psi_p\oplus\W_p\oplus(-L_p\V_p)\subseteq\Psi_p.
		\end{equation}
		Then $\Psi\coloneqq\Psi_p\times\{0\}\times\{0\}$ is RPI for the estimation error $\epsilon$.
	\end{lemma}
	\begin{proof}
		Define $\epsilon_p(k)\coloneqq[ I \; 0 \; 0 ]\epsilon(k)$ and suppose $\epsilon(k)\in\Psi$. With the dynamics of system \eqref{eq:system} and observer \eqref{eq:system_obs}, we have $\epsilon(k+1) = x(k+1)-\hat{x}(k+1) = f(x(k),u(k)) \hspace{-0.5pt}+\hspace{-0.5pt} w(k) - f(\hat{x}(k),u(k)) \hspace{-0.5pt}-\hspace{-0.5pt} B_\epsilon \epsilon(k) \hspace{-0.5pt}-\hspace{-0.5pt} B_v v(k) \stackrel[\beta(k)=\hat{\beta}(k)]{u_s(k)=\hat{u}_s(k)}{=} \left[\substack{ (A_p-L_pC_p) \epsilon_p(k)+w_p(k)-L_pv_p(k) \\ 0 \\ 0 }\right] \stackrel{\epsilon_p(k)\in\Psi_p}{\in} (A_p-L_pC_p) \Psi_p\oplus\W_p\oplus(-L_p\V_p) \times \{0\} \times\{0\} \stackrel{\eqref{eq:cond_RPI_obs}}{\subseteq} \Psi_p \times \{0\} \times\{0\} = \Psi$.
	\end{proof}
	\begin{remark}
		Since $(A_p,C_p)$ is detectable, $L_p$ can be chosen such that $A_p-L_pC_p$ is Schur. This is a sufficient condition for existence of $\Psi_p$ that satisfies \eqref{eq:cond_RPI_obs} (cf. \cite{kolmanovsky1998theory}).
	\end{remark}

	Hence, under condition \eqref{eq:cond_RPI_obs}, $\Psi$ is an RPI set for the estimation error $\epsilon$. It immediately follows that if $\epsilon(0)\in\Psi$, then $\epsilon(k)\in\Psi$ for all $k\in\N_0$. As a result, the ``observer disturbance'' $\delta(k)\coloneqq B_\epsilon \epsilon(k) + B_v v(k)$ satisfies $\delta(k)\in \Delta \coloneqq B_\epsilon \Psi + B_v \V=L_p(C_p\Psi_p\oplus\V_p)\times\{0\}\times\{0\}$ for all $k\in\N_0$.

	Next, we introduce the nominal system neglecting the disturbances
	\begin{equation} \label{eq:system_nom}
	\bar{x}(k+1) = f(\bar{x}(k),\bar{u}(k))
	\end{equation}
	with the nominal state $\bar{x}(k)\in\R^{n}$ and nominal input $\bar{u}(k)\in\R^{m}$, which will be used in the OCP for predictions. The fact that the estimation error does not leave $\Psi$ allows to bound its influence in the prediction. The idea is now, as in \cite{mayne2006robust_output}, to bound also the \emph{control error} between observer and nominal system $e(k)\coloneqq\hat{x}(k)-\bar{x}(k)$. This is done by appending the nominal input $\bar{u}(k)$, given by the solution to the OCP, by an error feedback which is then applied to the observer and the real system. As the observer state differs from the nominal state by $e$, and the real state differs from the observer state by $\epsilon$, it holds that $x(k)=\bar{x}(k)+e(k)+\epsilon(k)$. Therefore, bounding $e$, as done next, will enable us to tighten the constraint sets for the nominal state and input in the predictions in order to ensure robust constraint satisfaction of the real state and input.
	
	\subsection{Error feedback and bound for control error}
	
	Recall that in our NCS setup, updated inputs need to be sent to the actuator via a network, and that the transmission traffic is constrained by the token bucket TS. As a result, it is in particular not possible to transmit and apply an updated version of the error feedback in each time step, as is done in classical tube-based MPC (cf. \cite{mayne2006robust_output,Bayer14,Dong18_2}).
	
	In the following, we will present three different strategies to achieve an adequate error feedback in an NCS setup nonetheless, the applicability of which depends heavily on the capabilities of the actuator: For the first, called the \textit{ZOH actuator}, we assume that it has no processor or clock and can only hold the last received control update $u_c$. As a result, $u_e$ is chosen to $0$ and some kind of error feedback is already included in the control updates $u_c$. This is the simplest conceivable type of actuator, which requires no additional assumptions compared to disturbance-free rollout ETC (cf. \cite{Antunes14,Gommans17,Wildhagen19_2}). For the second, called the \textit{prediction-based actuator}, we assume that it is able to perform simple arithmetic operations in addition to the ZOH. Then, the error feedback can be implemented by predicting the observer state locally at the actuator in an auxiliary state variable $\tilde{x}$, and the error feedback can be chosen as\footnote{We implicitly assume here that a transmitted packet contains not only a control update $u_c$, but additionally the current nominal plant state $\bar{x}$ and the current observer state $\hat{x}$. In a packet-based network, the message part of a packet is typically much longer than what is required to transmit a single numerical value. It is thus reasonable to assume that in addition to the control update, the nominal and observer state can be transmitted in the same packet.} $u_e=K(\tilde{x}-\bar{x})$. The auxiliary system is different from the nominal one since the error feedback is appended to the nominal input in case of the former. In the third scenario, called the \textit{local measurement actuator}, we assume in addition that it has direct access to output measurements of the plant\footnote{\label{footnote:MEAS}One might raise the concern that the local measurement actuator could just run an MPC scheme like \cite{mayne2006robust_output} to steer the plant towards the origin under satisfaction of the constraints, since it has access to measurements and is collocated with the plant. However, we assume here that only the smart sensor possesses sufficient computational power to solve an optimization problem in real time. We note that this is computationally much more demanding than to perform some simple arithmetics. Further, it is not an alternative that the actuator implements a simple linear state feedback to drive the plant towards the origin, since in this case, constraint satisfaction could not be guaranteed.}. Then, it can run a copy of the Luenberger observer locally and the error feedback can simply be chosen as $u_e=K(\hat{x}-\bar{x})$.

	In what follows, we will bound the control error $e$ in an RCI set. The exact form of this RCI set will depend on the chosen actuator and error feedback. First, we will introduce a suitable mathematical definition of an error feedback and RCI set for NCSs. Then, we will describe the three different types of actuators and error feedbacks using this mathematical framework, and derive conditions for when a set is control invariant for each of these options.
	
	\subsubsection{Definition of $[1,H]$ control invariant set}
	
	For each of the three strategies, we need to take into account that updated information can not be sent to the actuator in each time step $k$ due to the constraints of the token bucket TS. For this reason, we need to derive a bound on the control error $e$ which is valid although the actuator did not receive any update of the error feedback or the nominal and observer state for several consecutive time steps. We say a set is $[1,H]$ robust control invariant (RCI), $H\in\N$, if it has the following property: If at $i=0$, the control error $e$ is contained in the set and a transmission takes place, and then the actuator receives no updated information for up to $H$ consecutive time steps, it does not leave the set in any of these time $H$ steps. Hence, knowledge of a $[1,H]$ RCI set enables us to bound the influence of the disturbance on the control error $e$.
	
	Next, we introduce a formal definition of a $[1,H]$ RCI set. To this end, consider an auxiliary system
	\begin{equation} \label{eq:system_aux}
	\tilde{x}(k+1) = f(\tilde{x}(k),u(k))
	\end{equation}
	to predict the observer state, with the auxiliary state $\tilde{x}(k)\in\R^{n}$. It receives the same input as the observer and real system, but neglects the disturbance. Further, consider that the control
	\begin{equation}
	\bar{u}(0)=\bar{\nu}(0)\coloneqq\begin{bmatrix} \bar{\nu}_c(0) \\ 1 \\ 0\end{bmatrix}\hspace{-2pt}
	, \:
	\bar{u}(i)=\bar{\nu}(i)\coloneqq\begin{bmatrix} 0 \\ 0 \\ 0 \end{bmatrix}\hspace{-2pt}, \: \forall i\in\N_{[1,H-1]} \label{eq:multi_step_control}
	\end{equation}
	is applied to the nominal system \eqref{eq:system_nom} with $\bar{\nu}_c(0)\in\U_p$. Then, the error feedback applied to the real \eqref{eq:system}, observer \eqref{eq:system_obs} and auxiliary system \eqref{eq:system_aux} is
	\begin{align}
	u(0)\hspace{-2pt}&=\hspace{-2pt}\phi'(\bar{\nu}(0),\hat{x}(0),\tilde{x}(0),\bar{x}(0))\hspace{-2pt}\coloneqq\hspace{-2pt}\begin{bmatrix} \mu(\bar{\nu}_c(0),\hat{x}_p(0),\bar{x}_p(0)) \\ 1 \\ \eta(\hat{x}_p(0),\tilde{x}_p(0),\bar{x}_p(0))  \end{bmatrix}\hspace{-2pt}, \nonumber \\
	u(i)\hspace{-2pt}&=\hspace{-2pt}\phi''(\bar{\nu}(i),\hat{x}(i),\tilde{x}(i),\bar{x}(i))\hspace{-2pt}\coloneqq\hspace{-2pt}\begin{bmatrix} 0\\ 0 \\ \eta(\hat{x}_p(i),\tilde{x}_p(i),\bar{x}_p(i))  \end{bmatrix}\hspace{-2pt},
	\label{eq:multi_step_error_feedback}
	\end{align}
	for all $i\in\N_{[1,H-1]}$, where $\mu:\U_p\times\X_p\times\X_p\to\U_p$ and $\eta:\X_p\times\X_p\times\X_p\to\U_p$. Note that the term $\mu$ may only depend on the nominal control update and on the observer and nominal state. As $\gamma(0)=\bar{\gamma}(0)=1$ and $\gamma(i)=\bar{\gamma}(i)=0$, $i\in\N_{[1,H-1]}$, $\mu$ is sent to the actuator at $i=0$ and then zero-order held according to the system dynamics \eqref{eq:system}. In contrast, the term $\eta$, which may depend on the observer, auxiliary and nominal state, can take different values in each time step. The reasoning behind this is the following: the term $\eta$ must not be used as a steering input, which must be computed by the rollout ETC and sent via the network, but it represents an error feedback which is locally computed at the actuator.
	
	Let us consider next the state sequences of the nominal system under \eqref{eq:multi_step_control} and of the observer and auxiliary system under \eqref{eq:multi_step_error_feedback}, i.e., that $\bar{u}(i)=\bar{\nu}(i)$, $u(0)=\phi'(\bar{\nu}(0),\hat{x}(0),\tilde{x}(0),\bar{x}(0))$ and $u(i)=\phi''(\bar{\nu}(i),\hat{x}(i),\tilde{x}(i),\bar{x}(i))$ for all $i\in\N_{[1,H-1]}$:
	\begin{align}
	\bar{x}(1) &= f(\bar{x}(0),\bar{\nu}(0)) \nonumber \\
	\tilde{x}(1) &= f(\tilde{x}(0),\phi'(\bar{\nu}(0),\hat{x}(0),\tilde{x}(0),\bar{x}(0))) \label{traj_1} \\
	\hat{x}(1) &= f(\hat{x}(0),\phi'(\bar{\nu}(0),\hat{x}(0),\tilde{x}(0),\bar{x}(0)))+ \delta(0) \nonumber
	\end{align}
	and for all $i\in\N_{[1,H-1]}$
	\begin{align}
	\bar{x}(i+1) &= f(\bar{x}(i),\bar{\nu}(i)) \nonumber \\
	\tilde{x}(i+1) &= f(\tilde{x}(i),\phi''(\bar{\nu}(i),\hat{x}(i),\tilde{x}(i),\bar{x}(i))) \label{traj_2} \\
	\hat{x}(i+1) &= f(\hat{x}(i),\phi''(\bar{\nu}(i),\hat{x}(i),\tilde{x}(i),\bar{x}(i)))+ \delta(i). \nonumber
	\end{align}
	The control error is then defined for all $i\in\N_{[0,H]}$ by
	\begin{equation*}
	e(i)\coloneqq \hat{x}(i)-\bar{x}(i),
	\end{equation*}
	where $\bar{x}(i)$, $\tilde{x}(i)$ and $\hat{x}(i)$ evolve according to \eqref{traj_1} and \eqref{traj_2}.
	
	\begin{definition} \label{def:RCI_set}
		For a given $H\in\N$, a set $\Omega\subseteq\R^{n}$ is $[1,H]$ RCI for the control error $e$ if and only if there exist $\mu:\U_p\times\X_p\times\X_p\to\U_p$ and $\eta:\X_p\times\X_p\times\X_p\to\U_p$ such that for all $\hat{x}(0),\tilde{x}(0),\bar{x}(0)\in\X$, $\tilde{x}(0)=\hat{x}(0)$ with $e(0)\in\Omega$, all $\bar{\nu}_c(0)\in\U_p$ and all $\{\delta(i)\}_{i=0}^{H-1}\in\Delta^H$ it holds that
		\begin{equation*}
		e(i) \in \Omega, \; \forall i\in\N_{[1,H]}.
		\end{equation*}
	\end{definition}

	According to Definition \ref{def:RCI_set}, a set is $[1,H]$ RCI if the control error $e$ does not leave the set for up to $H$ consecutive time steps under the error feedback $\phi',\phi''$. Next, we will present three strategies to choose the involved terms $\mu$ and $\eta$, depending on the capabilities of the actuator.
	
	\subsubsection{ZOH actuator} For the first strategy, we assume that the actuator has no computational abilities, clock or memory and can only hold the last received control update in a ZOH fashion. In this case, we propose to choose $\mu$ as the nominal control update appended by a linear feedback on the current control error. Further, $\eta$ is set to zero, since the actuator cannot perform any computations. In the following result, we formalize the proposed error feedback and derive a sufficient condition for when a set is $[1,H]$ RCI. The proof can be found in Appendix \ref{app:lem_RPI_set_ZOH}. We define $B_p^i\coloneqq\sum_{j=0}^{i-1}(A_p)^i B_p$.
	
	\begin{lemma} \label{lem:RPI_set_ZOH}
		Suppose there exist $H\in\N$, $K_p\in\R^{m_p\times n_p}$ and a set $\Omega_p\subseteq\R^{n_p}$, containing the origin, which satisfy
		\begin{equation}
		(A_p^i+B_p^i K_p)\Omega_p \oplus \bigg(\bigoplus_{j=0}^{i-1} A_p^i L_p (C_p\Psi_p\oplus\V_p)\bigg) \subseteq \Omega_p \label{eq:inv_ZOH}
		\end{equation}
		for all $i\in\N_{[1,H]}$. Then, with $\mu(\bar{\nu}_c,\hat{x}_p,\bar{x}_p)\coloneqq\bar{\nu}_c + K_p(\hat{x}_p-\bar{x}_p)$ and $\eta(\hat{x}_p,\tilde{x}_p,\bar{x}_p)\coloneqq0$, $\Omega\coloneqq\Omega_p\times K_p\Omega_p\times\{0\}$ is $[1,H]$ RCI for the control error $e$.
	\end{lemma}
	
	The particular form of Condition \eqref{eq:inv_ZOH} allows us, given $K_p$ and $H$, to use methods from the literature on RPI sets for linear difference inclusions to check existence of a suitable $\Omega_p$ \cite[Theorem 1]{Kouramas05}, and to construct $\Omega_p$ in practice \cite[Algorithm 1]{Kouramas05}, \cite[Algorithm 2]{Pluymers05}. A condition similar to \eqref{eq:inv_ZOH} appeared already in the context of tube-based self-triggered MPC \cite{Zhan17}, where knowledge of a suitable $K_p$ was assumed, but no systematic design procedure was given. In the next lemma, we address this issue by stating conditions that permit an efficient search for $K_p$. The proof can be found in Appendix \ref{app:RPI_set_ZOH_find_K}.
	\begin{lemma} \label{lem:RPI_set_ZOH_find_K}
		Suppose $\Psi_p$ from Lemma \ref{lem:RPI_obs} is a $\mathcal{C}$ set and there exist a scalar $\lambda\in(0,1)$ and matrices $X=X^\top\succeq 0\in\R^{n_p\times n_p}$ and $Y\in\R^{m_p\times n_p}$ such that
		\begin{equation}
		\begin{bmatrix}
		X & A_p^i X + B_p^i Y \\
		X A_p^{i\top} + Y^\top B_p^{i\top} & \lambda X
		\end{bmatrix} \succeq 0 \label{LMI_feedback_gain}
		\end{equation} 
		is satisfied for all $i\in\N_{[1,H]}$. Then, if $K_p\coloneqq YX^{-1}$, there exists an $\Omega_p$ such that Condition $\eqref{eq:inv_ZOH}$ is satisfied.
	\end{lemma}

	\begin{remark}
		For fixed $\lambda$, a simultaneous search for $X$ and $Y$ can be performed via a semi-definite program.
	\end{remark}

	\subsubsection{Prediction-based actuator} For the second strategy, we presume that the actuator possesses sufficient computational power to perform simple arithmetic operations, such that it may predict the auxiliary and nominal state in order to apply the error feedback. This is possible since it is known beforehand that in between transmission instants, the nominal input remains unchanged. For this strategy, we propose to set $\mu$ to be the nominal input and $\eta$ to be a linear error feedback, penalizing the difference between the auxiliary and nominal state variable. The following result establishes a condition under which a set is $[1,H]$ RCI under such an error feedback. The proof can be found in Appendix \ref{app:RPI_set_PRED}.
	
	\begin{lemma} \label{lem:RPI_set_PRED}
		Suppose there exist $H\in\N$, $K_p\in\R^{m_p\times n_p}$ and a set $\Omega_p\subseteq\R^{n_p}$, containing the origin, which satisfy
		\begin{equation}
		(A_p+B_p K_p)^i\Omega_p \oplus \bigg(\bigoplus_{j=0}^{i-1} A_p^i L_p (C_p\Psi_p\oplus\V_p)\bigg) \subseteq \Omega_p \label{eq:inv_PRED}
		\end{equation}
		for all $i\in\N_{[1,H]}$. Then, with $\mu(\bar{\nu}_c,\hat{x}_p,\bar{x}_p)\coloneqq\bar{\nu}_c$ and $\eta(\hat{x}_p,\tilde{x}_p,\bar{x}_p)\coloneqq K_p(\tilde{x}_p-\bar{x}_p)$, $\Omega\coloneqq\Omega_p\times \{0\}\times\{0\}$ is $[1,H]$ RCI for the control error $e$.
	\end{lemma}
	
	A similar condition for a $[1,H]$ RCI set also appeared in the context of tube-based self-triggered MPC in \cite{Brunner14}. Since $(A_p,B_p)$ is stabilizable, $K_p$ may be chosen such that $A_p+B_pK_p$ is Schur. By \cite[Lemma 2]{Brunner14}, it is immediate that this is a sufficient condition for existence of $\Omega_p$ that satisfies \eqref{eq:inv_PRED}. As in the case of the ZOH actuator, a suitable set $\Omega_p$ can be computed via \cite{Kouramas05} or \cite[Algorithm 2]{Pluymers05}.
	
	\subsubsection{Local measurement actuator} In addition to some computational abilities for arithmetic operations, we assume in the third scenario that the actuator has access to output measurements in each time step. Then, it may predict the nominal state and furthermore, run a local copy of the Luenberger observer to compute the observer state. For this strategy, we propose to choose $\mu$ as the nominal input and to penalize the difference between the observer and nominal state in $\eta$. In the next result, we establish a sufficient condition for a $[1,H]$ RCI set using this error feedback. The proof can be found in Appendix \ref{app:RPI_set_MEAS}.
	
	\begin{lemma} \label{lem:RPI_set_MEAS}
		Suppose there exist $K_p\in\R^{m_p\times n_p}$ and a set $\Omega_p\in\R^{n_p}$, containing the origin, which satisfy
		\begin{equation} \label{eq:inv_MEAS}
		(A_p+B_pK_p)\Omega_p \oplus L_p (C_p\Psi_p \oplus \V_p) \subseteq \Omega_p.
		\end{equation}
		Then, with $\mu(\bar{\nu}_c,\hat{x}_p,\bar{x}_p)\coloneqq\bar{\nu}_c$ and $\eta(\hat{x}_p,\tilde{x}_p,\bar{x}_p)\coloneqq K_p(\hat{x}_p-\bar{x}_p)$, $\Omega\coloneqq\Omega_p\times \{0\}\times\{0\}$ is $[1,H]$ RCI for the control error $e$ for any $H\in\N$.
	\end{lemma}

	In contrast to the ZOH and the prediction-based actuator, the RCI set for the local measurement actuator is independent of the open-loop length $H$. Furthermore, we note that \eqref{eq:inv_MEAS} is equivalent to the condition for when a set is RPI \cite{kolmanovsky1998theory}. As a result, standard methods for computing RPI sets can be used to construct a suitable $\Omega_p$ \cite{rakovic2005invariant}.
	
	\subsection{Bound for error between real and nominal system}
	
	Having bounded the estimation error as well as the control error, we are now able to bound the error between real and nominal state. As discussed above, this will make it possible to tighten the constraint sets for the nominal system used in the predictions, such that the real state will be guaranteed to satisfy the original constraints. In the following proposition, we establish that under any of the three error feedback strategies, the error between real and nominal state remains in the set $\Omega\oplus\Psi$ for at least $H$ time steps if $\phi'$ is applied if there is a transmission, and $\phi''$ is applied if not. The proof is straightforward and therefore omitted.
	
	\begin{proposition} \label{prop:boundedness_error}
		Suppose there exist $L_p,\Psi_p$ such that the conditions of Lemma \ref{lem:RPI_obs} are satisfied, and there exist $H,K_p,\Omega_p$ such that the conditions of either Lemma \ref{lem:RPI_set_ZOH}, \ref{lem:RPI_set_PRED} or \ref{lem:RPI_set_MEAS} are satisfied. In addition, suppose for some $t\in\N_0$, $x(t)$, $\hat{x}(t)$, and $\{\bar{u}(i)\}_{i=t}^{t+h-1}$, $h\in\N_{[1,H]}$, are given and the latter satisfies $\bar{\gamma}(t)=1$ and $\bar{\gamma}(i)=0$, $i\in\N_{[t+1,t+h-1]}$. Furthermore, suppose the initial states satisfy  $e(t)=\hat{x}(t)-\bar{x}(t)\in\Omega$, $\epsilon(t)=x(t)-\hat{x}(t)\in\Psi$ and $\tilde{x}(t) \coloneqq \hat{x}(t)$.
		
		Then, if it holds that
		\begin{equation*}
			u(i) \coloneqq \begin{cases} \phi'(\bar{u}(i),\hat{x}(i),\tilde{x}(i),\bar{x}(i)) & \bar{\gamma}(i) = 1 \\
			\phi''(\bar{u}(i),\hat{x}(i),\tilde{x}(i),\bar{x}(i)) & \bar{\gamma}(i) = 0 \\
			\end{cases},
		\end{equation*}
		$i\in\N_{[t,t+h-1]}$, with $\mu$ and $\eta$ either from Lemma \ref{lem:RPI_set_ZOH}, \ref{lem:RPI_set_PRED} or \ref{lem:RPI_set_MEAS}, then
		the observer state satisfies $\hat{x}(i)\in\{\bar{x}(i)\}\oplus\Omega$ and the real state satisfies $x(i)\in\{\hat{x}(i)\}\oplus\Psi\subseteq\{\bar{x}(i)\}\oplus\Omega\oplus\Psi$ for all $i\in\N_{[0,h]}$ and for any $\{w(i)\}_{i=t}^{t+h-1}\in \W^h$, $\{v(i)\}_{i=t}^{t+h-1}\hspace{-0.5 pt}\in\hspace{-0.5 pt}\V^h$.
	\end{proposition}
	
	\section{Robust Rollout ETC} \label{sec:MPC_scheme}
	
	In this section, we detail the proposed control scheme. First, we will introduce some preliminaries: The cyclic prediction horizon, the constraint tightening for the predictions, and a constraint on the predicted transmission schedule to enforce that inter-transmission intervals are at most $H$ time steps. Second, we will present the robust rollout ETC algorithm. Third, we will elaborate on its theoretical properties and establish recursive feasibility, robust constraint satisfaction and convergence.
	
	\subsection{Preliminaries: cyclic prediction horizon, constraint tightening and schedule constraint}
	
	Rollout ETC requires special care when designing the underlying OCP. Since the transmission of a control update is not possible in every time step, standard techniques from MPC to ensure recursive feasibility and convergence cannot be directly applied. To guarantee these properties nonetheless, we employ a cyclically time-varying prediction horizon \cite{Wildhagen19,Koegel13} in place of a traditional constant horizon. For a cycle length $M\in\N$ and a maximum horizon $\overline{N}\in\N_{\ge M}$, the prediction horizon at time $k$ is given by (see also \cite[Figure 1]{Wildhagen19})
	\begin{equation} \label{eq:shrinking_horizon}
	N(k) \coloneqq \overline{N}-(k\mod M).
	\end{equation}
	
	As asserted in Proposition \ref{prop:boundedness_error}, the error between real and nominal state is bounded by $\Omega\oplus\Psi$ regardless of the disturbance realization. Hence, we tighten the constraint sets according to
	\begin{align}
	\hat{\X} &\coloneqq (\X_p\ominus\Psi_p)\times (\U_p\ominus K_p\Omega_p) \times \N_{[0,b]}, \nonumber \\
	\bar{\X} &\coloneqq (\X_p\ominus\Omega_p\ominus\Psi_p)\times (\U_p\ominus K_p\Omega_p) \times \N_{[0,b]}, \label{eq:tightened_constraints} \\
	\bar{\U} &\coloneqq (\U_p\ominus K_p\Omega_p) \times \{0,1\} \times\{0\}, \nonumber
	\end{align}
	where $\bar{\X}$ and $\bar{\U}$ are to be used in the prediction.

	However, recall that the error between real and nominal state is only bounded if the inter-transmission interval is no longer than $H$ time steps. To ensure the latter in closed loop, we introduce a counter $s(k)$ to keep track of the time since the last transmission occurred. To be precise, the counter is defined such that at time $k$, the last transmission was at $k-s(k)-1$ (please see Algorithm \ref{scheme_MPC} for an explicit definition of $s(k)$). The idea is now to require that the predicted transmission schedules lie in a certain set, denoted by $\Gamma_N^H(s(k))$, as a constraint in the OCP in order to ensure a transmission at least every $H$ time steps in closed-loop. The set $\Gamma_N^H(s)$ is defined as the set of all schedules of length $N$, wherein the first transmission occurs after at most $H-s-1$ steps and in which the remaining ones are at most $H$ steps apart. If $N$ is smaller than $H-s$, the schedule needs not contain a transmission. In order to introduce a formal definition of this set, we define first the ordered sequence of indices for which a schedule $\gamma(\cdot)\in\{0,1\}^N$ is $1$, according to
	\begin{equation*}
	\tau_{\gamma(\cdot)}(\cdot) \coloneqq \{j\in\N_{[0,N-1]}|\gamma(j)=1\} \in\N_{[0,N-1]}^{n_{\gamma(\cdot)}},
	\end{equation*}
	where $n_{\gamma(\cdot)}\coloneqq\sum_{i=0}^{N-1}\gamma(i)$ is the number of transmissions in $\gamma(\cdot)$. Then, a formal definition of $\Gamma_N^H(s)$ is given by
	\begin{align}
	&\Gamma_N^H(s) \coloneqq \{\gamma(\cdot)\in\{0,1\}^N|(\tau_{\gamma(\cdot)}(0) \le H-s-1 \nonumber \\
	&\land \tau_{\gamma(\cdot)}(j+1)-\tau_{\gamma(\cdot)}(j) \le H, \; \forall j\in\N_{[0,n_{\gamma(\cdot)} - 2]} \label{def_scheduling_constraint} \\
	&\land N-\tau_{\gamma(\cdot)}(n_{\gamma(\cdot)}) \le H-1) \lor\; N\le H-s-1\}.  \nonumber
	\end{align}
	
	\subsection{Robust rollout ETC algorithm}
	
	In this subsection, we detail the robust rollout ETC algorithm. To this end, we introduce the predicted nominal state and input trajectories at time $k$ as $\bar{x}(\cdot|k)\coloneqq\{\bar{x}(0|k),$ $\ldots,\bar{x}(N(k)|k)\}$ and $\bar{u}(\cdot|k)\coloneqq\{\bar{u}(0|k),\ldots,\bar{u}(N(k)-1|k)\}$. The objective function is defined by
	\begin{align*}
	V(\bar{x}(\cdot|k),\bar{u}&(\cdot|k),k) \coloneqq \lambda(\bar{x}(0|k)) \\
	&+ \sum_{i=0}^{N(k)-1} \ell(\bar{x}(i|k),\bar{u}(i|k)) + V_f(\bar{x}(N(k)|k)),
	\end{align*}
	where $\lambda(\bar{x})\coloneqq\bar{u}_s^\top S \bar{u}_s$, $S\in\R^{m_p\times m_p}$, $R\succeq S \succ 0$ is an additional weight on the initial predicted state, and $V_f:\bar{\X}\to\R$ is the terminal cost. The stage cost is given by
	\begin{equation} \label{eq:stage_cost}
	\ell(\bar{x},\bar{u})\coloneqq \bar{x}_p^\top Q \bar{x}_p + (1-\bar{\gamma})\bar{u}_s^\top R \bar{u}_s + \gamma \bar{u}_c^\top R \bar{u}_c,
	\end{equation}
	which is the quadratic cost associated with the plant \eqref{eq:cost_plant}.
	\begin{remark}
		The term $\lambda$ is included in the objective function to establish convergence later, where the analysis relies on results from tube-based economic MPC \cite{Bayer14,Dong18_2}. It is indeed essential for this guarantee, since the stage cost \eqref{eq:stage_cost} is not positive definite w.r.t. $\bar{u}_s$. Note that here, the effect of $\lambda$ on the closed loop can be made negligibly small since $S$ can be chosen arbitrarily small as long as it is positive definite.
	\end{remark}
	
	Now, we are ready to state the OCP for robust rollout ETC. Given the observer state $\hat{x}(k)$, the nominal state $\bar{x}(k)$ and the counter $s(k)$, we define the mixed-integer quadratic program\footnote{The numerical complexity of solving \eqref{eq:MPC_OP} is discussed in \cite[Section I]{Wildhagen20}.} $\mathcal{P}(\hat{x}(k),\bar{x}(k),s(k),k)$, solved in each time step $k$, as
	\begin{subequations} \label{eq:MPC_OP}
		\begin{align}
		&\min_{\bar{x}(\cdot|k),\bar{u}(\cdot|k)} V(\bar{x}(\cdot|k),\bar{u}(\cdot|k),k)  \nonumber \\
		\text{s.t. } &\begin{cases}	\hat{x}(k)\in\{\bar{x}(0|k)\}\oplus\Omega & \bar{\gamma}(0|k) = 1 \label{constr_IC}\\
		\bar{x}(0|k)=\bar{x}(k) & \bar{\gamma}(0|k) = 0 \end{cases}  \\
		&\bar{x}(i+1|k) = f(\bar{x}(i|k),\bar{u}(i|k)) \label{constr_dynamics} \\
		&\bar{x}(i|k) \in\bar{\X}, \; \bar{u}(i|k) \in \bar{\U}, \quad \forall i\in\N_{[0,N(k)-1]} \label{constr_state_input}  \\
		&\bar{\gamma}(\cdot|k)\in\Gamma_{N(k)}^H(s(k)) \label{constr_scheduling} \\
		&\bar{x}(N(k)|k) \in \bar{\X}_f \label{constr_terminal}
		\end{align}
	\end{subequations}
	with the closed terminal region $\bar{\X}_f\subseteq\bar{\X}$. We comment on the exact choice of terminal ingredients to ensure recursive feasibility and convergence in Subsection \ref{sec:rec_feas_conv}.
	
	We denote by the superscript $^*$ the nominal state and input trajectories that solve $\mathcal{P}$. The robust rollout ETC operates according to Algorithm \ref{scheme_MPC}. Hence, the triggering and control law \eqref{eq:impl_trig_contr_law} is given by
	\begin{align*}
	[u_c(k)^\top \: \gamma(k)]^\top&=\kappa_{\text{RETC}}(y_p(k),k) \\
	&\coloneqq\begin{bmatrix} \bar{\gamma}^*(0|k)\mu(\bar{u}_c^*(0|k),\hat{x}_p(k),\bar{x}_p^*(0|k)) \\ \bar{\gamma}^*(0|k)
	\end{bmatrix},
	\end{align*}
	where $\bar{\gamma}^*(0|k),\bar{u}_c^*(0|k)$ and $\bar{x}_p^*(0|k)$ depend implicitly on $y_p(k)$ through the observer and the solution to the time-varying OCP \eqref{eq:MPC_OP}. Furthermore, the error feedback $u_e(k)\coloneqq\eta(\hat{x}_p(k),\tilde{x}_p(k),\bar{x}_p^*(0|k))$ is applied to the plant. The terms $\mu,\eta$ depend on the chosen actuator.

	\begin{algorithm}[h]
		\begin{Algorithm}\label{scheme_MPC}
			\normalfont{\textbf{Robust Rollout ETC Algorithm}}
			\begin{enumerate}
				\item[0)] Set $k=0$, enforce $\bar{\gamma}^*(0|0)=1$, choose an arbitrary $\hat{x}(0)$ that fulfills $x(0)\in\{\hat{x}(0)\}\oplus\Psi$, an $\bar{x}(0)$ that fulfills $\hat{x}(0)\in\{\bar{x}(0)\}\oplus\Omega$, and set $s(0)=0$.
				\item[1)] Solve $\mathcal{P}(\hat{x}(k),\bar{x}(k),s(k),k)$.
				\item[2)] Measure $y(k)$. Set $\bar{x}(k)\coloneqq\bar{x}^*(0|k)$ and apply the nominal input $\bar{u}(k)\coloneqq\bar{u}^*(0|k)$ to the nominal system \eqref{eq:system_nom}. \\
				\underline{Case $\bar{\gamma}^*(0|k)=1$:} Set $\tilde{x}(k)\coloneqq\hat{x}(k)$. Apply the error feedback
				\begin{equation*}
				\begin{aligned}
				u(k) &= \phi'(\bar{u}(k),\hat{x}(k),\tilde{x}(k),\bar{x}(k)) \\
				&= \left[\substack{\mu(\bar{u}_c^*(0|k),\hat{x}_p(k),\bar{x}_p^*(0|k))\\ 1 \\
					\eta(\hat{x}_p(k),\tilde{x}_p(k),\bar{x}_p^*(0|k))}\right]
				\end{aligned}
				\end{equation*}
				to the real system \eqref{eq:system}, the observer system \eqref{eq:system_obs} and \\ the auxiliary system \eqref{eq:system_aux}. Set $s(k+1)\coloneqq0$.
				\underline{Case $\bar{\gamma}^*(0|k)=0$:} Apply the error feedback
				\begin{equation*}
				\begin{aligned}
				u(k) &= \phi''(\bar{u}(k),\hat{x}(k),\tilde{x}(k),\bar{x}(k)) \\  &=\left[\substack{0\\ 0\\\eta(\hat{x}_p(k),\tilde{x}_p(k),\bar{x}_p^*(0|k)) }\right]
				\end{aligned}
				\end{equation*}
				to the real system \eqref{eq:system}, the observer system \eqref{eq:system_obs} and \\ the auxiliary system \eqref{eq:system_aux}. Set $s(k+1)\coloneqq s(k)+1$.
				\item[3)] Set $k\leftarrow k + 1$ and go to 1).
			\end{enumerate}
		\end{Algorithm}
	\end{algorithm}
	
	\begin{remark}
		In Step 2) of Algorithm \ref{scheme_MPC}, the sensor only needs to measure $y_p$, whereas $u_s$ and $\beta$ are internal variables. Likewise, $\hat{x},\tilde{x}$ and $s$ are internal variables of the controller and actuator, which justifies writing $\kappa_{\text{RETC}}$ as being dependent on the plant output $y_p$ and time $k$  only.
	\end{remark}
	\begin{remark}
		We enforce $\bar{\gamma}^*(0|0)=1$ to be able to give theoretical guarantees later. This requires $\beta_0\ge c-g$.
	\end{remark}
	
	 Note that according to \eqref{constr_IC}, the initial nominal state $\bar{x}(0|k)$ is an optimization variable only if there is a transmission predicted at initial time, otherwise it is fixed to $\bar{x}(0|k)=\bar{x}(k)=\bar{x}^*(1|k-1)$. This can be seen as a mixed strategy between \cite[Section 6]{Bayer14}, \cite{Dong18_2,mayne2006robust_output} where the initial state is always an optimization variable, and \cite[Section 3]{Bayer14} where the initial state is always fixed. The reason for this mixed strategy is the following: If one would reoptimize the nominal plant state although no transmission takes place, then one could also not send an updated version of the error feedback to the actuator. As a result, one could not guarantee that the control error is bounded by $\Omega$. However, if there is a transmission to the actuator scheduled, the nominal plant state can indeed be reoptimized since also the error feedback can be updated. We note that all theoretical guarantees (cf. Subsection \ref{sec:rec_feas_conv}) would be retained if one would enforce $\bar{x}(0|k)=\bar{x}(k)$ also at times when a transmission does take place. However, in this case there would be no feedback from the plant to the controller whatsoever. This is indeed provided by our mixed strategy, and typically results in better performance \cite{Bayer14}.

	\subsection{Theoretical properties: recursive feasibility, robust constraint satisfaction and convergence} \label{sec:rec_feas_conv}
	
	In this subsection, we derive conditions on the terminal region and cost such that recursive feasibility of the OCP, satisfaction of the original constraints and convergence to a region around the origin can be guaranteed. We design the terminal ingredients similarly as in \cite{Wildhagen19_2}, where the cycle length $M\coloneqq\lceil \frac{c}{g}\rceil$ is set to the base period of the token bucket TS.
	\begin{assumption} \label{ass:term_inv}
		There exist a closed set $\bar{\X}_{f,p} \subseteq \X_p\ominus\Omega_p\ominus\Psi_p$, containing the origin, and $K_{f,p}\in\R^{m_p\times n_p}$, such that $K_{f,p}\bar{\X}_{f,p}\subseteq\U_p\ominus K_p\Omega_p$, $(A_p^i+B_p^i K_{f,p})\bar{\X}_{f,p}\subseteq\X_p\ominus\Omega_p\ominus\Psi_p$ for all $i\in\N_{[1,M-1]}$ and $(A_p^M+B_p^M K_{f,p})\bar{\X}_{f,p}\subseteq\bar{\X}_{f,p}$.
	\end{assumption}
	\begin{assumption} \label{ass:term_decr}
		There exists a $P_{f,p}\succ 0\in\R^{n_p\times n_p}$ such that
		\begin{align*}
		&(A_p^M + B_p^M K_{f,p})^\top P_{f,p} (A_p^M + B_p^M K_{f,p}) - P_{f,p} \\
		&\le \hspace{-1.5pt}-\hspace{-2.8pt}\sum_{i=0}^{M-1}\hspace{-2pt} (A_p^i \hspace{-1pt}+\hspace{-1pt} B_p^i K_{f,p})^\top\hspace{-1pt} Q_p\hspace{-1pt} (A_p^i \hspace{-1pt}+\hspace{-1pt} B_p^i K_{f,p})\hspace{-1pt} -\hspace{-1pt} M K_{f,p}^\top R_p K_{f,p}.
		\end{align*}
	\end{assumption}
	\begin{remark}
		If the pair $(A_p^M,B_p^M)$ is controllable, an admissible choice for $K_{f,p}$ and $P_{f,p}$ is given in \cite[Lemmas 1,2]{Wildhagen19_2}. If $\X_p\ominus\Omega_p\ominus\Psi_p$ and $\U_p\ominus K_p\Omega_p$ are polytopic, a suitable set $\bar{\X}_{f,p}$ can be constructed using methods as in \cite[Remark 2]{Wildhagen20}. \label{rem_construction terminal_ingredients}
	\end{remark}
	
	With these two assumptions, we choose the terminal cost $V_f(\bar{x})\coloneqq \bar{x}_p^\top P_{f,p} \bar{x}_p$, the terminal region
	\begin{equation*}
	\bar{\X}_f \coloneqq \bar{\X}_{f,p} \times(\U_p\ominus K_p\Omega_p)\times\N_{[c-g,b]},	\end{equation*}
	and the terminal control sequence $\kappa_0(x) \coloneqq[(K_{f,p} \bar{x}_p)^\top 1]^\top$ and $\kappa_i(x) \coloneqq [0\; 0]^\top$ for all $j \in \N_{[1,M-1]}$. We note that $K_{f,p}$ can be different from $K_p$.
	
	\begin{assumption} \label{ass:RPI_sets}
		There exist $L_p,\Psi_p$ such that the conditions of Lemma \ref{lem:RPI_obs} are satisfied and $K_p,\Omega_p$ and $H\ge M$ such that the conditions of Lemmas \ref{lem:RPI_set_ZOH}, \ref{lem:RPI_set_PRED} or \ref{lem:RPI_set_MEAS} are satisfied. Also, $\X_p\ominus\Omega_p\ominus\Psi_p$ and $\U_p\ominus K_p\Omega_p$ are closed and contain the origin.
	\end{assumption}
	
	As an intermediate step towards recursive feasibility, the following result establishes that under application of Algorithm \ref{scheme_MPC}, there is indeed a transmission at least every $H$ time steps in closed loop. The proof can be found in Appendix \ref{app:Q_step_closed_loop}.
	
	\begin{lemma} \label{lem:H_step_closed_loop}
		Suppose $\overline{N}\ge H\ge M=\lceil \frac{c}{g}\rceil$ and for a $\overline{k}\in\N_{0}$, $\mathcal{P}(x(k),\bar{x}(k),s(k),k)$ is feasible for all $k\in\N_{[0,\overline{k}-1]}$. Then for any $k\in\N_{[0,\overline{k}-H]}$, it holds that $\sum_{i=k}^{k+H-1} \bar{\gamma}(i) \ge 1$.
	\end{lemma}

	\begin{remark} \label{rem:size_H}
		The exact choice of $H$ poses a tradeoff: by increasing $H$, longer inter-transmission intervals can be tolerated, which gives more flexibility for scheduling. However, doing so would also enlarge $\Omega_p$ for the ZOH and prediction-based actuators (see $\eqref{eq:inv_ZOH}$ and \eqref{eq:inv_PRED}), which would result in tighter constraints for the nominal plant (see \eqref{eq:tightened_constraints}). In effect, the feasible set (the set of states such that $\mathcal{P}$ is feasible) would shrink as well. In contrast, in case of the local measurement actuator,  the size of $\Omega_p$ is independent of $H$ (see \eqref{eq:inv_MEAS}).
	\end{remark}
	
	Finally, the following two results establish recursive feasiblity, robust constraint satisfaction and convergence. Note that constraint satisfaction also implies satisfaction of the token bucket TS and thus of the bandwidth constraint. The proofs can be found in Appendices \ref{app:rec_feas} and \ref{app:convergence}, respectively.
	
	\begin{theorem}[Recursive Feasibility and Robust Constraint Satisfaction] \label{thm:rec_feas}
		Suppose $\mathcal{P}(\hat{x}(0),\bar{x}(0),0,0)$ with the constraint $\bar{\gamma}^*(0|0)=1$ is feasible, Assumptions \ref{ass:term_inv} and \ref{ass:RPI_sets} hold, and $\overline{N}\ge H \ge M=\lceil \frac{c}{g}\rceil$. Then, $\mathcal{P}$ is feasible for all $k\in\N_{0}$, the observer state satisfies $\hat{x}(k)\in\hat{\X}$ for all $k\in\N_{0}$, and the real state and input satisfy $x(k)\in\X$ and $u(k)\in\U$ for all $k\in\N_{0}$.
	\end{theorem}

	\begin{theorem}[Convergence] \label{thm:convergence}
		Suppose the conditions of Theorem \ref{thm:rec_feas}, Assumpt. \ref{ass:term_decr} and $R\succeq S \succ0$ hold. Then, the nominal state $\bar{x}(k)$ converges to $\mathbb{A}\coloneqq\{0\}\times\{0\}\times\N_{[0,b]}$, the observer state $\hat{x}(k)$ converges to $\mathbb{A}\oplus\Omega$, and the real state $x(k)$ converges to $\mathbb{A}\oplus\Omega\oplus\Psi=(\Omega_p\oplus\Psi_p)\times K_p\Omega_p \times \N_{[0,b]}$ as $k\to\infty$.
	\end{theorem}

	\begin{remark}
		In case of the local measurement actuator, the guarantees of Theorems \ref{thm:rec_feas} and \ref{thm:convergence} hold even if the schedule constraint \eqref{constr_scheduling} was omitted and there was no condition placed on $H$. This is immediate in view of Lemma \ref{lem:RPI_set_MEAS}.
	\end{remark}
	
	\section{Numerical example} \label{sec:num_ex}
	
	We consider a disturbed double integrator
	\begin{align*}
	x_p(k+1) &= \begin{bmatrix} 1 & 0.1 \\ 0 & 1 \end{bmatrix} x_p(k) + \begin{bmatrix} 0.005 \\ 0.1 \end{bmatrix} u_p(k) + w_p(k) \\
	y_p(k) &= \begin{bmatrix} 1 & 0	\end{bmatrix} x_p(k) + v_p(k)
	\end{align*}
	subject to the constraints $x_p(k)\in[-20,20]^2$, $u_p(k)\in[-20,20]$. The bounded disturbance and measurement noise fulfill $w_p(k)\in[-0.002,0.002]^2$ and $v_p(k)\in[-0.001,0.001]$. We presume that the network prescribes a bandwidth limit of $\frac{1}{3}$, which can be captured by a token bucket TS with the parameters $g=1$, $c=3$ and $b=10$ such that $\frac{g}{c}=\frac{1}{3}$. As a result, the base period is $M=3$ and $H\ge M=3$ must be satisfied. The numerical results were obtained using Matlab, Yalmip \cite{YALMIP}, SDPT3 \cite{SDTP3} and MPT3 \cite{MPT3}.
	
	As mentioned in Section \ref{sec:bounding}, the construction of $[1,H]$ RCI sets is possible via \cite{Kouramas05},\cite{Pluymers05}. In this example, we obtained better results with the algorithm stated below. A proof that the resulting $\Omega_p$ fulfills \eqref{eq:inv_ZOH} or \eqref{eq:inv_PRED} is straightforward.
	\begin{enumerate}
		\item Find a contractive set $\mathbb{E}\subseteq\R^{n_p}$ via $\min_{\lambda_i\in(0,1]} \lambda_i$ s.t. $A_{K}^i\mathbb{E}\subseteq\lambda_i\mathbb{E}$ for all $i\in\N_{[1,H]}$, where $A_{K}^i\coloneqq A_p^i+B_p^iK$ for the ZOH actuator or $A_{K}^i\coloneqq (A_p+B_pK)^i$ for the prediction-based actuator, e.g., via \cite[Algorithm 2]{Pluymers05}.
		\item Solve $\min_{\delta_i\in[0,\infty)} \delta_i$ s.t. $\sum_{j=0}^{i-1} A_p^jL_p(C_p\Psi_p\oplus\V_p)\subseteq\delta_i\mathbb{E}$ for all $i\in\N_{[1,H]}$.
		\item Compute $\rho \coloneqq \max_{i\in\N_{[1,H]}} \frac{\delta_i}{1-\lambda_i}$ and define $\Omega_p\coloneqq\rho\mathbb{E}$.
	\end{enumerate}

	\begin{table}
		\centering
		\caption{Area of $[1,H]$ RCI set $\Omega_p$ for different types of actuator and different maximum inter-transmission intervals $H$ (area of $\Omega_p+\Psi_p$ in brackets).}
		\begin{tabular}{c|ccc}
			$H$ & ZOH & prediction-based & local measurement \\ \hline
			3 & 12.0 (12.9) & 6.7 (7.3) & 2.1 (2.5) \\
			4 & 22.4 (23.6) & 13.2 (14.1) & 2.1 (2.5) \\
			5 & 35.2 (36.6) & 22.7 (23.9) & 2.1 (2.5) \\
			6 & 71.1 (73.5) & 36.1 (37.6) & 2.1 (2.5) 
		\end{tabular}
	\label{tab:RPI_sets_H}
	\end{table}
	
	Table \ref{tab:RPI_sets_H} lists the areas of $\Omega_p$ for all three types of actuator (ZOH, prediction-based and local measurement) and for $H=3,4,5,6$. First, we study the effect of the maximum allowed inter-transmission interval $H$ on the size of the $[1,H]$ RCI set $\Omega$. It is apparent that the RCI set grows considerably with the open-loop phase $H$ for the ZOH and prediction-based actuator. In this example for the ZOH actuator, the area of $\Omega_p$ is 5.9 times larger for $H=6$ compared to $H=3$. This shows again the tradeoff that was discussed in Remark \ref{rem:size_H}: if we choose $H=3$, which is the same as the cycle length, then the schedule constraint \eqref{constr_scheduling} will enforce an almost periodic transmission pattern in closed loop. If one wants more flexibility for scheduling and thus increases $H$, one pays with a larger $[1,H]$ RCI set, which leads to tighter constraints in the prediction and in turn to a smaller feasible set. In contrast, for the local measurement actuator, the size of the RCI set is independent of $H$. This demonstrates that if the actuator has access to plant measurements in order to apply the error feedback, this tradeoff vanishes.
	
	Second, we compare the area of the $[1,H]$ RCI set for the different actuators. For any $H$, the set $\Omega_p$ is largest for the ZOH actuator, while it is about half the size for the prediction-based actuator. For the local measurement actuator, it is always much smaller than in the previous cases. As can be seen from Theorem \ref{thm:convergence} and as discussed in Remark \ref{rem:size_H}, a smaller $[1,H]$ RCI set brings two decisive advantages: a smaller guaranteed region of convergence and a larger feasible set. Thus, we note that the higher requirements on the hardware posed by the prediction-based and the local measurement actuator are rewarded by important practical benefits.
	
	Table \ref{tab:RPI_sets_H} also lists the area of $\Omega_p+\Psi_p$. As $\Psi_p$ has a area of only $0.004$ and is independent of $H$ and the actuator, $\Omega_p$ contributes the main share of area for all tested cases.
	
	\begin{figure}
		\centering
		\input{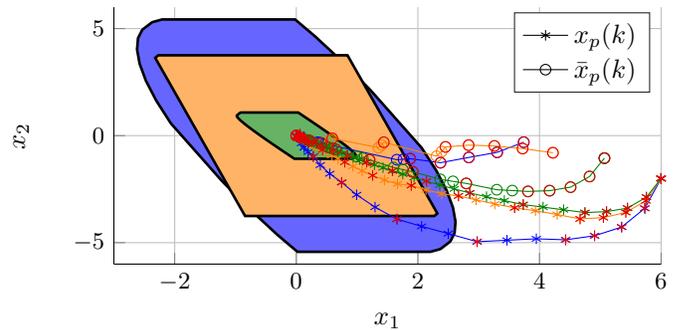} \vspace{-15pt}
		\caption{Trajectories of closed-loop system and $\Omega_p+\Psi_p$ for ZOH (blue), prediction-based (orange) and local measurement actuator (green). Transmission time points are marked in red.}
		\label{fig:plant_states}
	\end{figure}
	
	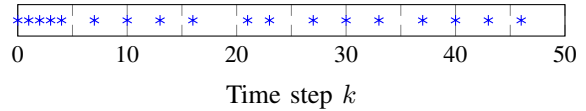
\begin{figure}
		\centering
		% This file was created by matlab2tikz.
%
%The latest updates can be retrieved from
%  http://www.mathworks.com/matlabcentral/fileexchange/22022-matlab2tikz-matlab2tikz
%where you can also make suggestions and rate matlab2tikz.
%
\begin{tikzpicture}

\begin{axis}[%
width=\columnwidth,
height=2cm,
xmin=0,
xmax=50,
ytick = \empty,
xtick = {0,5,10,15,20,25,30,35,40,45,50},
xticklabels = {0,,10,,20,,30,,40,,50},
ymin=-1,
ymax=1,
axis background/.style={fill=white},
xminorticks = true,
grid = minor,
xminorgrids = true,
xlabel = {Time step $k$},
]
\addplot [only marks, mark=asterisk, mark options={solid, istblue}, forget plot]
  table[row sep=crcr]{%
	0	0\\
	1	0\\
	2	0\\
	3	0\\
	4	0\\
	7	0\\
	10	0\\
	13	0\\
	16	0\\
	21	0\\
	23	0\\
	27	0\\
	30	0\\
	33	0\\
	37	0\\
	40	0\\
	43	0\\
	46	0\\
};
\end{axis}
\end{tikzpicture}%
 \vspace{-5pt}
		\caption{Transmission instants for ZOH actuator.}
		\label{fig:transmissions}
	\end{figure}

	A closed-loop simulation for all three actuators was conducted with a maximum prediction horizon of $\overline{N}=6$, $H=5$, cost matrices $Q=10I$, $R=1$, $S=10^{-6}$, and initial conditions $x_p(0)=[6\;-2]^\top$, $u_s(0)=0$ and $\beta(0)=b=10$, and extremal values for the disturbance and measurement noise. A suitable terminal region and cost were constructed as discussed in Remark \ref{rem_construction terminal_ingredients}. Figure \ref{fig:plant_states} shows the closed-loop trajectories of the plant and $\Omega_p+\Psi_p$ for all three actuators. One can clearly see that the true plant state converges to $\Omega_p+\Psi_p$ in all cases. Actually, the region of ultimate convergence is much smaller, which shows that the theoretical statement of Theorem \ref{thm:convergence} can be rather conservative. Figure \ref{fig:transmissions} displays at which time points a transmission takes place for the ZOH actuator. Due to the schedule constraint, the inter-transmission interval is indeed never longer than $H=5$. For all three types of actuator, $18$ transmissions were triggered in the time interval $k\in[0,50]$.
	
	\section{Summary and Outlook} \label{sec:summary}
	
	In this article, we designed a rollout ETC which is able to handle bounded uncertainties and output feedback. The implicit triggering and control law defined by the solution to an OCP made it possible to guarantee satisfaction of bandwidth constraints, which is a hard problem in classical ETC.
	
	Since the idea of rollout ETC is rather new, there are many important open issues: Transmission scheduling over both the controller-actuator and the sensor-controller channel, and distributed controllers/schedulers would allow for more flexible NCS architectures. Being able to handle delays, dropouts and quantization, which were not considered in this work for simplicity, would make rollout ETC even more suitable to address the special needs of NCSs. Lastly, generalizing the concepts for robustness, e.g., by investigating the $\ell_p$ gain, input-to-state stability or robustness for nonlinear systems, could be an interesting direction in rollout ETC.
	
	\section*{Acknowledgment}
	
	The authors thank Florian David Brunner for helpful discussions on the algorithm for constructing the $[1,H]$ RCI set.
	
	\bibliographystyle{IEEEtran}
	\bibliography{bib_Rollout}
	
	\appendix
	\section{Appendix}
	
	\subsection{Proof of Lemma \ref{lem:RPI_set_ZOH}} \label{app:lem_RPI_set_ZOH}
	
	We prove first that if for $i\in\N_{[0,H-1]}$ it holds that $e(i)\in\R^{n_p+m_p}\times\{0\}$, then with $\small\mathcal{A}_{\bar{\gamma}} \coloneqq {\footnotesize\begin{bmatrix}	A_p+\bar{\gamma}B_pK_p & (1-\bar{\gamma})B_p & 0 \\
		\bar{\gamma}K_p & (1-\bar{\gamma})I & 0 \\
		0 & 0 & 0	\end{bmatrix}}$
	\begin{equation}
	e(i+1)=
	\begin{cases}
	\mathcal{A}_{1}e(0)+\delta(0) & i=0 \\
	\mathcal{A}_{0}e(i)+\delta(i) & i\in\N_{[1,H]}
	\end{cases}. \label{eq:error_dyn}
	\end{equation}
	
	Consider first $i=0$, for which we have $e(1) = f(\hat{x}(0),\left[\substack{\bar{\nu}_c(0)+K_p(\hat{x}_p(0)-\bar{x}_p(0))\\ 0 \\ 1}\right])+\delta(0)-f(\bar{x}(0),\left[\substack{\bar{\nu}_c(0) \\ 0 \\ 1}\right])$ $ \stackrel{\hat{\beta}(0)=\bar{\beta}(0)}{=}\left[\substack{	A_p(\hat{x}_p(0)-\bar{x}_p(0))+B_pK_p(\hat{x}_p(0)-\bar{x}_p(0)) \\
	K_p(\hat{x}_p(0)-\bar{x}_p(0))  \\
	0	}\right]
	+\delta(0)$.
	
	For $i\in\N_{[1,H-1]}$, we have again $\hat{\beta}(i)-\bar{\beta}(i) = 0$ and hence
	$e(i+1) = f(\hat{x}(i),\left[\substack{0\\ 0 \\0}\right])+\delta(i)-f(\bar{x}(i),\left[\substack{0\\ 0\\0}\right])=\left[\substack{	A_p(\hat{x}_p(i)-\bar{x}_p(i))+B_p(\hat{u}_s(i)-\bar{u}_s(i)) \\
	\hat{u}_s(i)-\bar{u}_s(i)  \\
	0	}\right] + \delta(i)$.
	
	Second, we prove by induction that if $e(0)\in\Omega$, then $e(i)= \mathcal{A}_0^{i-1} \mathcal{A}_1 e(0)+ \sum_{j=0}^{i-1} \mathcal{A}_0^j \delta(i-j-1)$ and $e(i)\in\Omega$ for all $i\in\N_{[1,H]}$, such that the claim follows directly. We begin the induction for $i=1$. Defining $e_p(i)\coloneqq\begin{bmatrix} I & 0 & 0 \end{bmatrix}e(i)$ for all $i\in\N_{[0,H]}$, we have with $e(0)\in\Omega\subseteq\R^{n_p+m_p}\times\{0\}$ and \eqref{eq:error_dyn} $e(1) = \mathcal{A}_1 e(0) + \delta(0) =
	\left[\substack{	(A_p+B_pK_p) e_p(0) \\ K_p e_p(0) \\ 0	}\right] + \delta(0) \stackrel[\delta(0)\in\Delta]{e_p(0)\in\Omega_p}{\in} (A_p+B_pK_p)\Omega_p \oplus L_p(C_p\Psi_p\oplus\V_p) \times K_p\Omega_p \times \{0\} \stackrel{\eqref{eq:inv_ZOH}}{\subseteq} \Omega$.
	
	Next, we consider the induction step. With \eqref{eq:error_dyn} and the induction hypothesis, we immediately have the result since $e(i+1) = \mathcal{A}_0 e(i) + \delta(i) = \mathcal{A}_0 \big(\mathcal{A}_0^{i-1} \mathcal{A}_1 e(0)+ \sum_{j=0}^{i-1} \mathcal{A}_0^j \delta(i-j-1)\big) + \delta(i)$. Further, using the definition of $B_p^i$, $e(i+1) = \left[\substack{	(A_p^{i+1}+B_p^{i+1} K_p) e_p(0) \\ K_p e_p(0) \\ 0	}\right] \hspace{-2pt}+\hspace{-2pt} \sum_{j=0}^{i} \mathcal{A}_0^j \delta(i-j) \stackrel[\delta(i-j)\in\Delta]{e_p(0)\in\Omega_p}{\in} (A_p^{i+1}\hspace{-1pt}+\hspace{-1pt}B_p^{i+1}$ $K_p)\Omega_p \oplus \big(\bigoplus_{j=0}^{i} A_p^j L_p (C_p\Psi_p\oplus\V_p)\big) \times K_p\Omega_p \times \{0\} \stackrel{\eqref{eq:inv_ZOH}}{\subseteq} \Omega$.

	\subsection{Proof of Lemma \ref{lem:RPI_set_ZOH_find_K}} \label{app:RPI_set_ZOH_find_K}
	
	We interpret $\{A_p^i+B_p^i K_p| i\in\N_{[1,H]}\}$ as the set of matrices defining a convex hull in which the system matrix of a linear difference inclusion is guaranteed to lie (see \cite{Kouramas05}).
	
	First, we note that $\bigoplus_{j=0}^{H-1} A_p^j L_p (C_p\Psi_p\oplus\V_p)$ is a $\mathcal{C}$ set, since matrix multiplication and Minkowski addition preserve convexity and compactness of $\Psi_p$ and $\V_p$. Hence, \cite[Assumption 1]{Kouramas05} is fulfilled.
	
	Second, we apply the Schur complement to \eqref{LMI_feedback_gain}, substitute $X=P^{-1}$ and $Y=K_pP^{-1}$, pre- and postmultiply $P$ and lastly, subtract $P$ on both sides to obtain $(A_p^i+B_p^i K_p)^\top P (A_p^i+B_p^i K_p) - P \preceq -(1-\lambda)P, \; i\in\N_{[1,H]}$,
	which implies \cite[Assumption 2]{Kouramas05}.
	
	Finally, we apply \cite[Theorem 1]{Kouramas05} to conclude that there exists a set $\Omega_p$ which fulfills $(A_p^i+B_p^i K_p)\Omega_p \oplus \bigg(\bigoplus_{j=0}^{H-1} A_p^j L_p (C_p\Psi_p\oplus\V_p)\bigg) \subseteq \Omega_p$
	for all $i\in\N_{[1,H]}$. The claim follows as $\bigoplus_{j=0}^{i-1} A_p^j L_p (C_p\Psi_p\oplus\V_p)\subseteq \bigoplus_{j=0}^{H-1} A_p^j L_p (C_p\Psi_p\oplus\V_p)$ for all $i\in\N_{[1,H]}$.
	
	\subsection{Proof of Lemma \ref{lem:RPI_set_PRED}} \label{app:RPI_set_PRED}
	
	Let us define $e_1(i)\coloneqq\hat{x}(i)-\tilde{x}(i)$ for all $i\in\N_{[1,H]}$. We will first prove that whenever $\hat{x}(0)=\tilde{x}(0)$, it holds that $e_1(i)\in\bigoplus_{j=0}^{i-1} A_p^j L_p (C_p\Psi_p\oplus\V_p)\times\{0\}\times\{0\}$ for all $i\in\N_{[1,H]}$.
	
	We start the induction for $i=1$. We have $e_1(1) = f(\hat{x}(0),\left[\substack{\bar{\nu}_c(0) \\ K_p(\tilde{x}_p(0)-\bar{x}_p(0)) \\ 1}\right]) + \delta(0)	- f(\tilde{x}(0),\left[\substack{\bar{\nu}_c(0) \\ K_p(\tilde{x}_p(0)-\bar{x}_p(0)) \\ 1}\right])$ $= \delta(0) \in \Delta = L_p (C_p\Psi_p\oplus\V_p)\times\{0\}\times\{0\}$.
	
	For the induction step, for any $i\in\N_{[1,H-1]}$ we have
	$e_1(i+1) = f(\hat{x}(i),\left[\substack{0 \\ K_p(\tilde{x}_p(i)-\bar{x}_p(i)) \\ 0}\right]) + \delta(i) - f(\tilde{x}(i),\left[\substack{0 \\ K_p(\tilde{x}_p(i)-\bar{x}_p(i)) \\ 0}\right]) \stackrel[\hat{\beta}(i)=\tilde{\beta}(i)]{\hat{u}_s(i)=\tilde{u}_s(i)}{=}
	\left[\substack{ A_p (\hat{x}_p(i)-\tilde{x}_p(i)) \\ 0 \\ 0 }\right] + \delta(i) \in \bigoplus_{j=1}^{i} A_p^j L_p (C_p\Psi_p\oplus\V_p)\oplus L_p (C_p\Psi_p\oplus\V_p)\times\{0\}\times\{0\}$
	where the last two steps are from the induction hypothesis.
	
	Further, we define $e_2(i) \coloneqq \tilde{x}(i) - \bar{x}(i)$. We prove that if $\hat{x}(0) = \tilde{x}(0)$, $\hat{x}(0)-\bar{x}(0)\in\Omega$, it holds that $e_2(i) = \text{diag}\{(A_p+B_pK_p)^i,0,0\}e_2(0)$ and $e_2(i)\in (A_p+B_pK_p)^i\Omega_p\times\{0\}\times\{0\}$ for all $i\in\N_{[1,H]}$. We start the induction with $e_2(1) = f(\tilde{x}(0),\left[\substack{\bar{\nu}_c(0) \\ K_p(\tilde{x}_p(0)-\bar{x}_p(0)) \\ 1}\right]) - f(\bar{x}(0),\left[\substack{\bar{\nu}_c(0) \\ 0 \\ 1}\right]) \stackrel[\tilde{\beta}(i)=\bar{\beta}(i)]{\tilde{u}_s(i)=\bar{u}_s(i)}{=}
	\left[\substack{ (A_p+B_pK_p) (\tilde{x}_p(0)-\bar{x}_p(0)) \\ 0 \\ 0 }\right] \in (A_p+B_pK_p)\Omega_p \times\{0\}\times\{0\}$.
	
	For the induction step, for any $i\in\N_{[1,H-1]}$ we compute
	$e_2(i+1) = f(\tilde{x}(i),\left[\substack{0 \\ K_p(\tilde{x}_p(i)-\bar{x}_p(i)) \\ 0}\right]) - f(\bar{x}(i),\left[\substack{0 \\ 0 \\ 0}\right]) \stackrel[\tilde{\beta}(i)=\bar{\beta}(i)]{\tilde{u}_s(i)=\bar{u}_s(i)}{=}
	\left[\substack{ (A_p+B_pK_p) (\tilde{x}_p(i)-\bar{x}_p(i)) \\ 0 \\ 0 }\right] \in (A_p+B_pK_p)^{i+1}\Omega_p \times\{0\}\times\{0\}$,
	where we used the induction hypothesis for the last two steps.
	
	Finally, since $e(i)\hspace{-1pt}=\hspace{-1pt}e_1(i) \hspace{-0.5pt}+\hspace{-0.5pt} e_2(i)$ for all $i\in\N_{[1,H]}$, we have $e(i) \in (A_p+B_p K_p)^i\Omega_p \oplus \bigg(\bigoplus_{j=0}^{i-1} A_p^i L_p (C_p\Psi_p\oplus\V_p)\bigg) \\ 
	\times\{0\}\times\{0\}\stackrel{\eqref{eq:inv_PRED}}{\subseteq} \Omega_p\times\{0\}\times\{0\} = \Omega$.
	
	\subsection{Proof of Lemma \ref{lem:RPI_set_MEAS}} \label{app:RPI_set_MEAS}
	
	Suppose $e(i)\in\Omega$ for an $i\in\N_0$. In case $i=0$, it holds that $e(1) = f(\hat{x}(0),\left[\substack{\bar{\nu}_c(0)\\ K_p(\hat{x}_p(0)-\bar{x}_p(0)) \\ 1}\right]+\delta(0)-f(\bar{x}(0),\left[\substack{\bar{\nu}_c(0) \\ 0 \\ 1}\right]) \stackrel[\hat{\beta}(i)=\bar{\beta}(i)]{\hat{u}_s(i)=\bar{u}_s(i)}{=}
	\text{diag}\{A_p+B_pK_p,0,0\} e(0) + \delta(0)$.
	
	If $i\in\N$, we have $e(i+1) = f(\hat{x}(i),\left[\substack{0\\ K_p(\hat{x}_p(i)-\bar{x}_p(i)) \\ 0}\right])+\delta(i)-f(\bar{x}(i),\left[\substack{0 \\ 0 \\ 0}\right]) \stackrel[\hat{\beta}(i)=\bar{\beta}(i)]{\hat{u}_s(i)=\bar{u}_s(i)}{=}
	\text{diag}\{A_p\hspace{-1pt}+\hspace{-1pt}B_pK_p,0,0\} e(i) + \delta(i)$.
	
	In both cases, we directly obtain $e(i\hspace{-1pt}+\hspace{-1pt}1) \in (A_p+B_p$ $K_p)\Omega_p \oplus L_p (C_p\Psi_p \oplus \V_p)\times\{0\}\times\{0\} \stackrel{\eqref{eq:inv_MEAS}}{\subseteq} \Omega_p \times\{0\}\times\{0\} = \Omega$.
	
	\subsection{Proof of Lemma \ref{lem:H_step_closed_loop}} \label{app:Q_step_closed_loop}
	
	Suppose, for contradiction, that for an arbitrary $k\in\N_{[0,\overline{k}-H]}$, $\bar{\gamma}(k)=1$ and $\bar{\gamma}(k+i)=0$ for all $i\in\N_{[1,H]}$ under application of Algorithm \ref{scheme_MPC}, i.e., that the inter-transmission interval is longer than $H$ although the predicted schedules always fulfill the constraint \eqref{constr_scheduling}.
	
	In this case, $s(k+i) = i-1$, $i\in\N_{[1,H]}$. The interval until the next time instant at which the full horizon $\overline{N}$ is used, is denoted by $h$, i.e., $N(k+h)=\overline{N}$. According to the cyclic horizon scheme, $h\in\N_{[1,M]}\subseteq \N_{[1,H]}$, where the inclusion is from $H\ge M$. Hence, $\bar{\gamma}^*(l|k+h)=1$ for some $l \in\N_{[0,H-s(k+h)-1]}=\N_{[0,H-h]}$ since $\bar{\gamma}^*(\cdot|k+h)\in\Gamma_{\overline{N}}^H(h-1)$. Note that since $\overline{N}\ge H$, it cannot happen that $N(k+h)=\overline{N}\le H-s(k+h)-1$ due to $s(k+h)=h-1\ge 0$ and in effect, $\bar{\gamma}^*(\cdot|k+h)$ must contain a transmission (cf. \eqref{def_scheduling_constraint}).
	
	Due to the cyclic horizon, $N(k+i) \ge \overline{N}-i+h$ for all $i\in\N_{[h,H]}$. Equality holds if the horizon just decreases in the remaining time while the strict inequality comes into effect if the horizon recovers to $\overline{N}$ at some point. Hence,
	\begin{equation*}
	N(k+i)\ge \overline{N}-i+h \hspace*{-2pt}\stackrel{\overline{N}\ge H}{\ge}\hspace*{-2pt} H-i+h \hspace*{-2pt}\stackrel{h\ge 1}{\ge}\hspace*{-2pt} H-i+1=H-s(k+i).
	\end{equation*}
	In consequence, since $\bar{\gamma}^*(\cdot|k+i)\in\Gamma_{N(k+i)}^H(i-1)$, it must hold that $\bar{\gamma}^*(l|k+i)=1$ for some $l\in\N_{[0,H-s(k+i)-1]}=\N_{[0,H-i]}$ for all $i\in\N_{[h,H]}$. Finally, it follows for some $i\in\N_{[h,H]}$ that $\bar{\gamma}^*(0|k+i)=1$ and then according to Algorithm \ref{scheme_MPC}, $\bar{\gamma}(k+i)=1$. This contradicts the premise, such that the statement holds.

	\subsection{Proof of Theorem \ref{thm:rec_feas}} \label{app:rec_feas}
	
	Suppose $\mathcal{P}$ was feasible at all instants up to time $k\in\N_{0}$. Consider the candidate initial condition at $k+1$
	\begin{equation}
	\bar{x}(0|k+1) = \bar{x}^*(1|k) = f(\bar{x}(k),\bar{u}(k)). \label{feas_IC}
	\end{equation}
	Define $t\coloneqq k+1-s(k+1)-1$. According to Algorithm \ref{scheme_MPC}, it holds that $\tilde{x}(t)=\hat{x}(t)$, $\hat{x}(t)-\bar{x}(t)\in\Omega$ and $u(t)=\phi'(\bar{u}(t),\hat{x}(t),\tilde{x}(t),\bar{x}(t))$. Furthermore, it holds that $u(t+i)=\phi''(\bar{u}(t+i),\hat{x}(t+i),\tilde{x}(t+i),\bar{x}(t+i))$, $i\in\N_{[1,s(k+1)]}$. Recall that Lemma \ref{lem:H_step_closed_loop} guarantees that there is a transmission at least every $H$ time steps, i.e., that $s(k+1)\le H-1$. Furthermore, Assumption \ref{ass:RPI_sets} ensures that $L_p,\Psi_p,K_p,\Omega_p,H$, as demanded in the prerequisites of Proposition \ref{prop:boundedness_error}, exist. Combining these arguments, we note that all conditions of Proposition \ref{prop:boundedness_error} are fulfilled with $t = k+1-s(k+1)-1$ and $h\coloneqq s(k+1)+1\in\N_{[1,H]}$, which ensures $\hat{x}(k+1)\in\{\bar{x}(0|k+1)\}\oplus\Omega$. In conclusion, the constraint \eqref{constr_IC} is fulfilled by \eqref{feas_IC}.
	
	Now consider the case $(k+1)\text{mod} M\neq 0$ where $N(k+1)=N(k)-1$, for which the candidate input sequence is
	\begin{equation}
	\bar{u}(\cdot | k + 1)=\{\bar{u}^*(1|k),\ldots,\bar{u}^*(N(k)-1|k)\}, \label{feas_input_1}
	\end{equation}
	and the case $(k+1)\text{mod} M=0$, where $N(k+1)=\overline{N}=N(k)+M-1$ and for which the candidate input is
	\begin{align}
	&\bar{u}(\cdot | k + 1)=\{\bar{u}^*(1|k),\ldots,\bar{u}^*(N(k)-1|k), \label{feas_input_2} \\
	&\kappa_{0}(\bar{x}(\overline{N}-M-1|k+1)),\ldots,\kappa_{M-1}(\bar{x}(\overline{N}-1|k+1))\}. \nonumber
	\end{align}
	We refer to the proof of \cite[Theorem 1]{Wildhagen19_2} and \cite{Koegel13} to conclude that $\bar{u}(\cdot | k + 1)$ fulfills \eqref{constr_dynamics}, \eqref{constr_state_input} and \eqref{constr_terminal} in both cases.
	
	Next, we focus on the schedule constraint \eqref{constr_scheduling}. We will show only that the candidate sequences \eqref{feas_input_1} and \eqref{feas_input_2} fulfill the schedule constraint of the ``first'' and ``last'' transmission in \eqref{def_scheduling_constraint}, since it is obvious that the ``middle transmissions'' are just shifted in \eqref{feas_input_1} and \eqref{feas_input_2} compared to $\bar{\gamma}^*(\cdot|k)$, and are still no more than $H$ time steps apart. Consider 1) $N(k)\ge H-s(k)$. Then, $\bar{\gamma}^*(i|k)=1$ for an $i\in\N_{[0,H-s(k)-1]}$ and $\bar{\gamma}^*(j|k)=1$ for a $j\in\N_{[N(k)-H,N(k)-1]}$. In case $\bar{\gamma}^*(0|k)=0$, both candidate input sequences \eqref{feas_input_1} and \eqref{feas_input_2} guarantee that $\bar{\gamma}(i|k+1)=1$ for an $i\in\N_{[0,H-s(k)-2]}=\N_{[0,H-s(k+1)-1]}$ since then, $s(k+1)=s(k)+1$ and $s(k)\ge 0$ and the previously optimal schedule was just shifted in the candidate schedules. In case where $\bar{\gamma}^*(0|k)=1$, we can conclude the same but forgo a derivation due to space constraints. Consider 1a) that  $(k+1)\text{mod}M\neq 0$. With \eqref{feas_input_1}, if $\tau_{\bar{\gamma}^*(\cdot|k)}(n_{\bar{\gamma}^*(\cdot|k)})\ge 1$, then $\bar{\gamma}(j|k+1)=1$ for a $j\in\N_{[N(k)-H-1,N(k)-2]}=\N_{[N(k+1)-H,N(k+1)-1]}$ due to $N(k+1)=N(k)-1$. If $\tau_{\bar{\gamma}^*(\cdot|k)}(n_{\bar{\gamma}^*(\cdot|k)})=0$, one may conclude that $N(k+1)\le H-s(k+1)-1$, such that it must not contain a transmission. In conclusion, $\bar{\gamma}(\cdot|k+1)\in\Gamma_{N(k)-1}^H(s(k+1))$. In case 1b) $(k+1)\text{mod}M=0$, we have $N(k+1)=\overline{N}=N(k)+M-1$ and $\bar{\gamma}(\overline{N}-M|k+1)=1$ due to the scheduled transmission in $\kappa_0$. Since $H\ge M$, we have again $\bar{\gamma}(\cdot|k+1)\in\Gamma_{N(k)+M-1}^H(s(k+1))$. Consider 2) that $N(k)\le H-s(k)-1$, such that $\bar{\gamma}^*(\cdot|k)$ might not contain a transmission at all. If 2a) $(k+1)\text{mod}M\neq 0$, then $N(k+1)=N(k)-1$, such that $N(k+1)=N(k)-1\le H-s(k)-2 \le H-s(k+1)-1$ since $s(k+1)\in\{0,s(k)+1\}$ and $s(k)\ge 0$. In effect, $\bar{\gamma}(\cdot|k+1)$ is not required to contain a transmission and hence, $\bar{\gamma}(\cdot|k+1)\in\Gamma_{N(k)-1}^H(s(k+1))$. If 2b) $(k+1)\text{mod}M = 0$, then $\bar{\gamma}(\overline{N}-M|k+1)=1$ due to $\kappa_0$. It holds on the one hand that $\bar{\gamma}(i|k+1)=1$ for an $i\in\{\overline{N}-M\}=\{N(k)-1\}\stackrel{N(k)\ge 1}{\subseteq}\N_{[0,N(k)-1]}\stackrel{N(k)\le H-s(k)-1}{\subseteq}\N_{[0,H-s(k-2)]}\subseteq\N_{[0,H-s(k+1)-1]}$, where the last inclusion is again from $s(k+1)\in\{0,s(k)+1\}$. On the other hand, $\bar{\gamma}(j|k+1)=1$ for a $j\in\{\overline{N}-M\}\subseteq\N_{[\overline{N}-H,\overline{N}-1]}$ since $H\ge M$. In conclusion, we have $\bar{\gamma}(\cdot|k+1)\in\Gamma_{N(k)+M-1}^H(s(k+1))$ for this last case as well.
	
	So far, we have proven feasibility of $\mathcal{P}(x(k),\bar{x}(k),s(k),k)$ and that $\hat{x}(k)\in\{\bar{x}(k)\}\oplus\Omega$, $x(k)\in\{\bar{x}(k)\}\oplus\Omega\oplus\Psi$ for all $k\in\N_{0}$. Finally, we elaborate on robust constraint satisfaction: First, let us investigate the input constraint $u(k)\in\U\coloneqq\{u_c,u_e\in\R^m|u_c+u_e\in\U_p\} \times \{0,1\}$. It is clear that $\gamma(k)\in\{0,1\}$, such that it remains to show that $u_c(k)+u_e(k)\in\U_p$. In case of the ZOH actuator, it holds that $u_c(k) = \bar{u}_c(k) + K_p(\hat{x}_p(k)-\bar{x}_p(k))$ and $u_e(k)=0$ whenever $\bar{\gamma}(k)=1$. It is immediate that $u_c(k)+u_e(k)\in\bar{\U} \oplus K_p\Omega_p = \U_p \ominus K_p\Omega_p \oplus K_p\Omega_p\subseteq\U_p$ from Assumption \ref{ass:RPI_sets}, such that the original constraint is fulfilled. If $\bar{\gamma}(k)=0$, $u_c(k)=u_e(k)=0$, such that the input constraints are trivially fulfilled. In case of the prediction-based actuator, we have $u_c(k) = \bar{u}_c(k)$ and $u_e(k) = K_p(\tilde{x}_p(k)-\bar{x}_p(k))$ whenever $\bar{\gamma}(k)=1$. From the proof of Lemma \ref{lem:RPI_set_PRED} and $\hat{x}(k)\in\bar{x}(k)\oplus\Omega$ for all $k\in\N_{0}$, it follows directly that $\tilde{x}_p(k)-\bar{x}_p(k)\in\Omega_p$ for all $k\in\N_{0}$, such that we finally conclude $u_c(k)+u_e(k)\in\U_p$. If $\bar{\gamma}(k)=0$, then $u_c(k)=0$ and $u_e(k)=K_p(\tilde{x}_p(k)-\bar{x}_p(k))$ such that $u_c(k)+u_e(k)\in\Omega_p$ is fulfilled due to Assumption \ref{ass:RPI_sets}. In case of the local measurement actuator, we can argue analogously. Second, it remains to show that the original state constraints are satisfied, which follows directly from the fact that $\bar{x}(i)\in\bar{\X}$, the definition of the tightened constraint sets \eqref{eq:tightened_constraints} and Assumption \ref{ass:RPI_sets}.
	
	\subsection{Proof of Theorem \ref{thm:convergence}} \label{app:convergence}
	
	We analyze convergence with the help of so-called rotated cost functions known from economic MPC \cite{Bayer14,Dong18_2}. We refer first to the proof of \cite[Theorem 1]{Wildhagen19_2} to conclude that \cite[Assumption 1]{Wildhagen19} is fulfilled for the nominal system \eqref{eq:system_nom} with the control invariant set $\mathbb{A}$, $\ell_{av}^*=0$ and any storage $\lambda(\bar{x})=\bar{u}_s^\top S \bar{u}_s$ for which $S$ fulfills $R\succeq S\succ0$. It is also proven therein that \cite[Assumptions 2-4]{Wildhagen19} are fulfilled as well. Define $L(\bar{x},\bar{u}) = \ell(\bar{x},\bar{u}) + \lambda(\bar{x})-\lambda(f(\bar{x},\bar{u}))-\ell_{av}^*$, $\tilde{V}_f(\bar{x})=V_f(\bar{x})+\lambda(\bar{x})$, the rotated objective function
	\begin{equation*}
	\tilde{V}(\bar{x}(\cdot|k),\bar{u}(\cdot|k),k) \hspace{-2pt} \coloneqq \hspace{-8pt} \sum_{i=0}^{N(k)-1} \hspace{-8pt} L(\bar{x}(i|k),\bar{u}(i|k)) + \tilde{V}_f(\bar{x}(N(k)|k))
	\end{equation*}
	and the rotated optimization problem $\tilde{\mathcal{P}}(\hat{x}(k),\bar{x}(k),s(k),k)$:
	\begin{equation*}
	\min_{\bar{x}(\cdot|k),\bar{u}(\cdot|k)} \tilde{V}(\bar{x}(\cdot|k),\bar{u}(\cdot|k),k) \text{ s.t. } \eqref{constr_IC}-\eqref{constr_terminal}.
	\end{equation*}
	
	The crucial observation is now that since $\lambda$ is a storage of the nominal system, $\mathcal{P}$ and $\tilde{\mathcal{P}}$ have the same optimizer (cf. \cite{Dong18_2}). In effect, we may prove convergence of the nominal state to $\mathbb{A}$ using the value function of $\tilde{\mathcal{P}}$, which goes along the same lines as the proof of \cite[Theorem 1]{Wildhagen19}.
	
	The theorem's second and third statements follow immediately with $\hat{x}(k)\in\{\bar{x}(k)\}\oplus\Omega$ and $x(k)\in\{\bar{x}(k)\}\oplus\Omega\oplus\Psi$ for all $k\in\N_{0}$ as established in the proof of Theorem \ref{thm:rec_feas}.
	
	\begin{IEEEbiography}[{\includegraphics[width=1in,height=1.25in,clip,keepaspectratio]{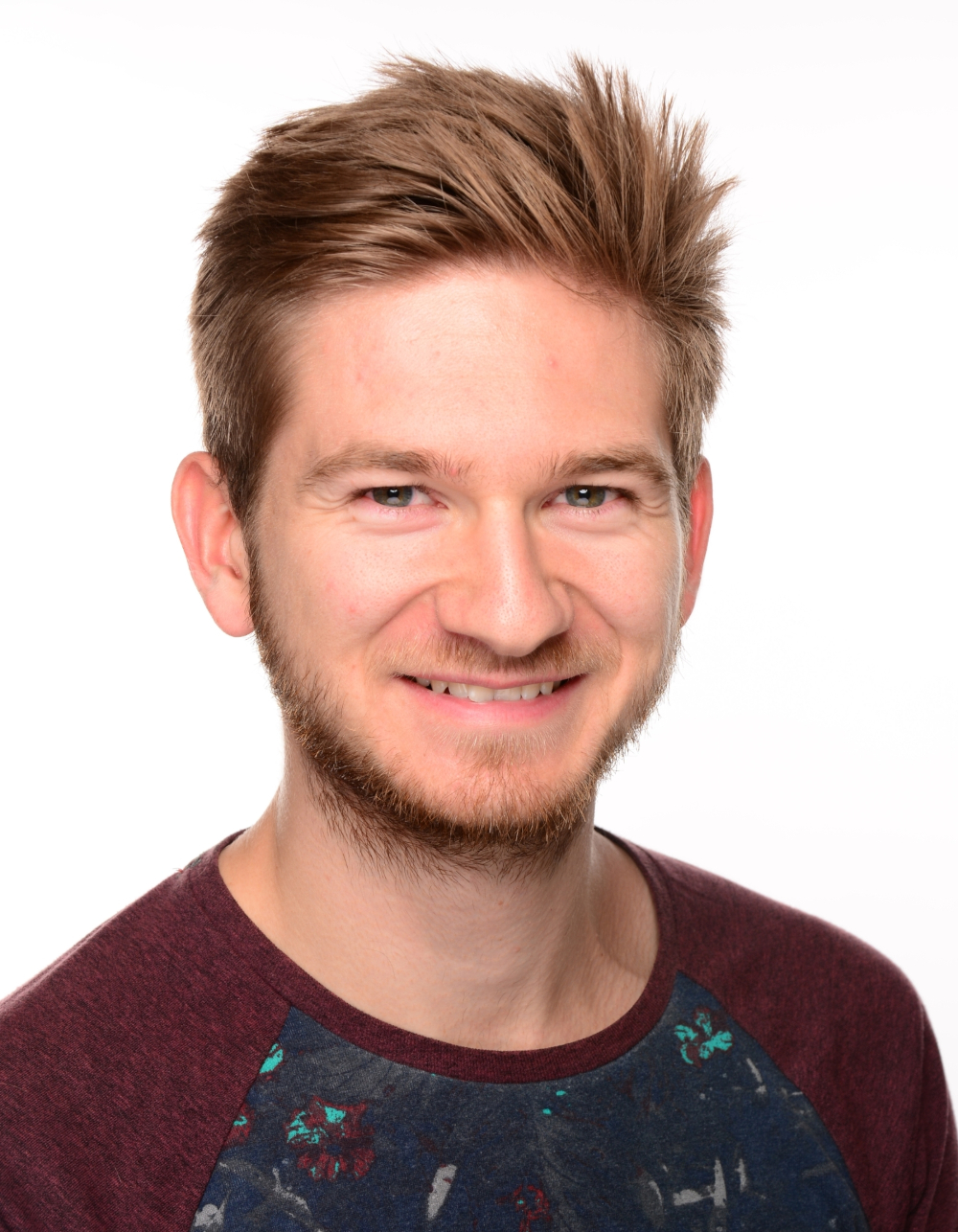}}]{Stefan Wildhagen} received the Master’s degree in Engineering Cybernetics from the University of Stuttgart, Germany, in 2018. He has since been a doctoral student at the Institute for Systems	Theory and Automatic Control under supervision of Prof. Allg\"ower and a member of the Graduate School Simulation Technology at the University of Stuttgart. His research interests are in the area of Networked Control Systems, with a focus on optimization-based scheduling and control as well as on data-driven methods.
	\end{IEEEbiography}

	\begin{IEEEbiography}[{\includegraphics[width=1in,height=1.25in,clip,keepaspectratio]{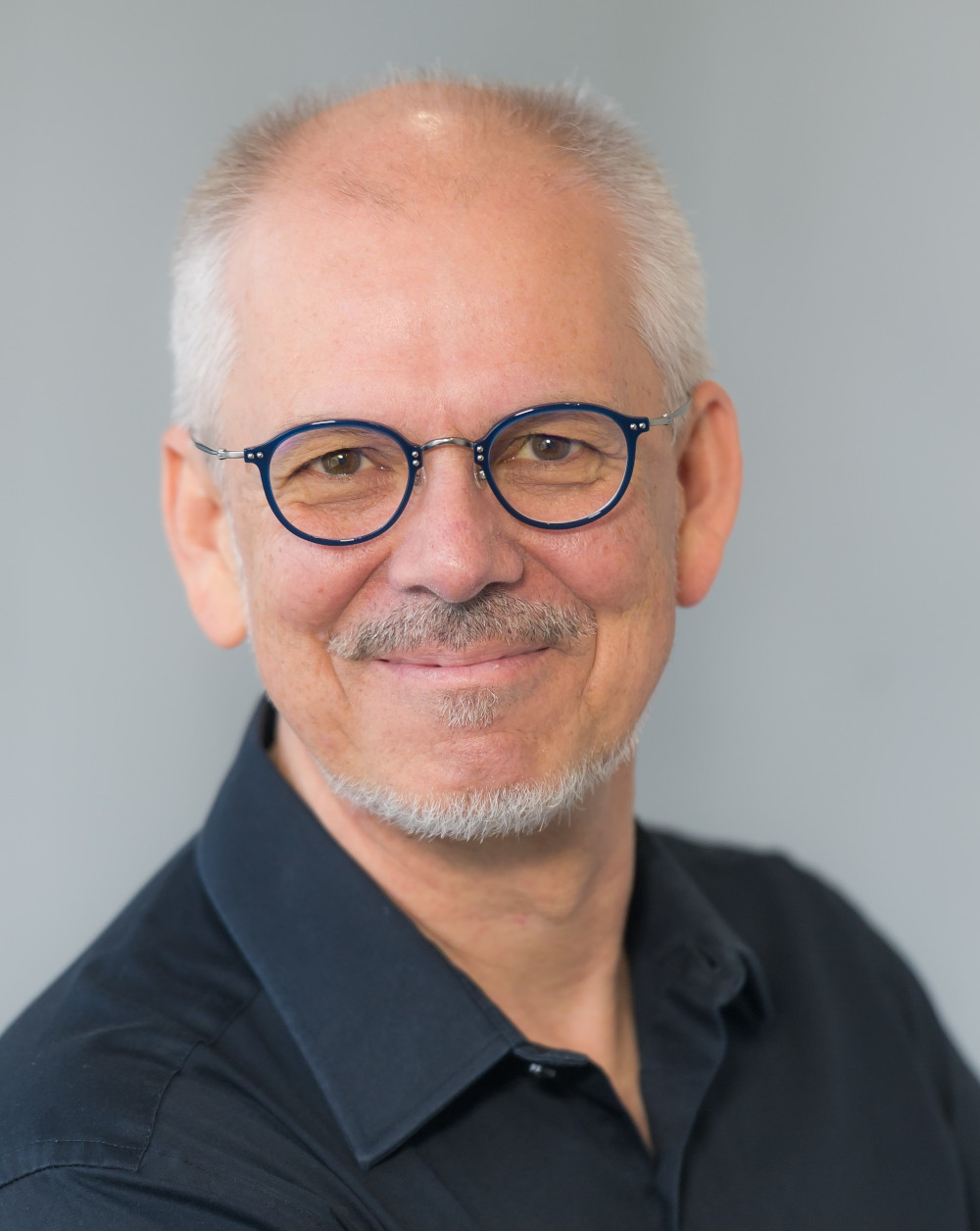}}]{Frank Allg\"ower} is professor of mechanical engineering at the University of Stuttgart, Germany, and Director of the Institute for Systems Theory and Automatic Control (IST) there.
		
		He is active in serving the community in several roles: Among others he has been President of the International Federation of Automatic Control (IFAC) for the years 2017-2020, Vicepresident for Technical Activities of the IEEE Control Systems Society for 2013/14, and Editor of the journal Automatica from 2001 until 2015. From 2012 until 2020 he served in addition as Vice-president for the German Research Foundation (DFG), which is Germany’s most important research funding organization.
		
		His research interests include predictive control, data-based control, networked control, cooperative control, and nonlinear control with application	to a wide range of fields including systems biology.
	\end{IEEEbiography}

\end{document}